\documentclass[entropy,article,accept,moreauthors,dvipdfm,12pt,a4paper]{mdpi}
\usepackage{makeidx}
\usepackage{graphicx}
\usepackage{amssymb,amsmath,amsthm}
\usepackage{hyphenat}
\usepackage[sort&compress]{natbib}
\usepackage{hypernat}

\input epsf
\newcommand{\D}{{\mathrm{d}}}
\newtheorem{theorem}{Theorem}
\newtheorem{lemma}{Lemma}
\newtheorem{proposition}{Proposition}
\newtheorem{definition}{Definition}
\newtheorem{corollary}{Corollary}

 \lastpage{1193} \doinum{10.3390/e12051145}
\pubvolume{12} \pubyear{2010} \history{Received: 1 March 2010; in revised form: 30 April 2010 / Accepted: 4 May 2010 /
\linebreak Published: 7 May 2010 / Corrected Postprint 9 November 2013}

\Title{Entropy: The Markov Ordering Approach}

\Author{Alexander N. Gorban $^{1,\star}$, Pavel A. Gorban $^{2}$ and
George Judge $^{3}$}

\address{$^{1}$Department of Mathematics, University of Leicester,
Leicester, UK\\ $^{2}$Institute of Space and Information
Technologies, Siberian Federal University, Krasnoyarsk, Russia\\
$^{3}$Department of Resource Economics, University of California,
Berkeley, CA, USA}

\corres{E-mail: ag153@le.ac.uk.}

\abstract{The focus of this article is on entropy and Markov
processes. We study the properties of functionals which are
invariant with respect to monotonic transformations and analyze two
invariant ``additivity'' properties: (i) existence of a monotonic
transformation which makes the functional additive with respect to
the joining of independent systems and (ii) existence of a monotonic
transformation which makes the functional additive with respect to
the partitioning of the space of states. All Lyapunov functionals
for Markov chains which have properties (i) and (ii) are derived. We
describe the most general ordering of the distribution space, with
respect to which all continuous-time Markov processes are monotonic
(the {\em Markov order}). The solution differs significantly from
the ordering given by the inequality of entropy growth. For
inference, this approach results in a convex compact set of
conditionally ``most random'' distributions.}


\begin{document}

\section{Introduction}
\subsection{A Bit of History: Classical Entropy}

Two functions, energy and entropy, rule the Universe.

In 1865 R. Clausius formulated two main laws \cite{Clausius1865}:
\begin{enumerate}
\item{The energy of the Universe is constant.}
\item{The entropy of the Universe tends to a maximum.}
\end{enumerate}

The universe is isolated. For non-isolated systems energy and
entropy can enter and leave, the change in energy is equal to its
income minus its outcome, and the change in entropy is equal to
entropy production inside the system plus its income minus outcome.
The entropy production is always positive.

Entropy was born as a daughter of energy. If a body gets heat $
\Delta Q $ at the temperature $T$ then for this body $\D S= \Delta
Q /T$. The total entropy is the sum of entropies of all bodies.
Heat goes from hot to cold bodies, and the total change of entropy
is always positive.

Ten years later J.W. Gibbs \cite{Gibbs1875} developed a general
theory of equilibrium of complex media based on the entropy maximum:
the equilibrium is the point of the conditional entropy maximum
under given values of conserved quantities. The entropy maximum
principle was applied to many physical and chemical problems. At the
same time J.W. Gibbs mentioned that entropy maximizers under a given
energy are energy minimizers under a given entropy.

The classical expression $\int p\ln p$ became famous in 1872 when
L. Boltzmann proved his \linebreak $H$-theorem \cite{Boltzmann1872}: the
function $$H=\int f(x,v) \ln f(x,v) \D x \D v$$ decreases in time
for isolated gas which satisfies the Boltzmann equation (here
$f(x,v)$ is the distribution density of particles in phase space,
$x$ is the position of a particle, $v$ is velocity). The
statistical entropy was born: $S=-kH$. This was the one-particle
entropy of a many-particle system (gas).

In 1902, J.W. Gibbs  published a book ``Elementary principles in
statistical dynamics'' \cite{Gibbs1902}. He considered ensembles in
the many-particle phase space with probability density
$\rho(p_1,q_1,\ldots p_n,q_n)$, where $p_i,q_i$ are the momentum and
coordinate of the $i$th particle. For this distribution,
\begin{equation}\label{GibbsEnt}
S=-k \int \rho(p_1,q_1,\ldots p_n,q_n) \ln
(\rho(p_1,q_1,\ldots p_n,q_n))
 \D q_1\ldots \D q_n \D p_1\ldots \D p_n
\end{equation}
Gibbs introduced the canonical distribution that provides the
entropy maximum for a given expectation of energy and gave rise to
the entropy maximum principle (MaxEnt).

The Boltzmann period of history was carefully studied
\cite{Villani}. The difference between the Boltzmann entropy which
is defined for coarse-grained distribution and increases in time
due to gas dynamics, and the Gibbs entropy, which is constant due
to dynamics, was analyzed by many authors
\cite{Jaynes1965,GoldsteinLebov2004}. Recently, the idea of two
functions, energy and entropy which rule the Universe was
implemented as a basis of two-generator formalism of
nonequilibrium thermodynamics \cite{Grmela1997,Ottinger2005}.

In information theory, R.V.L. Hartley (1928) \cite{Hartley1928}
introduced a logarithmic measure of information in electronic
communication in order ``to eliminate the psychological factors
involved and to establish a measure of information in terms of
purely physical quantities''. He defined information in a text of
length $n$ in alphabet of s symbols as $H=n \log s$.

In 1948, C.E. Shannon \cite{Shannon1948} generalized the Hartley
approach and developed ``a mathematical theory of communication'',
that is information theory. He measured information, choice and
uncertainty by \linebreak the entropy:
\begin{equation}\label{Shannon}
S=-\sum_{i=1}^n p_i \log p_i
\end{equation}

Here, $p_i$ are the probabilities of a full set of $n$ events
($\sum_{i=1}^n p_i=1$). The quantity $S$ is used to measure of how
much ``choice'' is involved in the selection of the event or of
how uncertain we are of the outcome. Shannon mentioned that this
quantity form will be recognized as that of entropy, as defined in
certain formulations of statistical mechanics. The classical
entropy (\ref{GibbsEnt}), (\ref{Shannon}) was called the
Boltzmann--Gibbs--Shannon entropy (BGS entropy). (In 1948, Shannon
used the {\it concave} function (\ref{Shannon}), but under the
same notation $H$ as for the Boltzmann {\it convex} function. Here
we use $H$ for the convex $H$-function, and $S$ for the concave
entropy.)

In 1951, S. Kullback and R.A. Leibler \cite{KullLei1951}
supplemented the BGS entropy by the relative BGS entropy, or the
Kullback--Leibler divergence between the current distribution $P$
and some ``base'' (or ``reference'') distribution $Q$:

\begin{equation}D_{\mathrm{KL}}(P\|Q) = \sum_i p_i \log \frac{p_i}{q_i}
\end{equation}

The Kullback--Leibler divergence is always non-negative
$D_{\mathrm{KL}}(P\|Q) \geq 0$ (the Gibbs inequality). It is not
widely known that this ``distance'' has a very clear physical
interpretation. This function has been well known in physical
thermodynamics since 19th century under different name. If $Q$ is an
equilibrium distribution at the same temperature as $P$ has, then
\begin{equation}\label{KLF}
D_{\mathrm{KL}}(P\|Q)=\frac{F(P)-F(Q)}{kT}
\end{equation}
where $F$ is free energy and $T$ is thermodynamic temperature. In
physics, $F=U-TS$, where physical entropy $S$ includes an
additional multiplier $k$, the Boltzmann constant. The
thermodynamic potential $-F/T$ has its own name, Massieu function.
Let us demonstrate this interpretation of the Kullback--Leibler
divergence. The equilibrium distribution $Q$ provides the
conditional entropy (\ref{Shannon}) maximum under a given
expectation of energy $\sum_i u_i q_i =U$ and the normalization
condition $\sum_i q_i =1$. With the Lagrange multipliers $\mu_U$
and $\mu_0$ we get the equilibrium Boltzmann distribution:

\begin{equation}\label{boltzmannDistribution}
q_i=\exp (-\mu_0 - \mu_U u_i)=\frac{\exp(-\mu_U u_i)}{\sum_i
\exp(-\mu_U u_i)}
\end{equation}
 The Lagrange multiplier $\mu_U$ is in physics (by definition) $1/kT$, so
$S(Q) = \mu_0+\frac{U}{kT}$, hence, $\mu_0=-\frac{F(Q)}{kT}$. For
the Kullback--Leibler divergence this formula gives  (\ref{KLF}).

After the classical work of Zeldovich (1938, reprinted in 1996
\cite{Zeld}), the expression for free energy in the
``Kullback--Leibler form" $$F=kT\sum_i
c_i\left(\ln\left(\frac{c_i}{c_i^*(T)}\right) -1\right)$$ where
$c_i$ is concentration and $c^*_i(T)$ is the equilibrium
concentration of the $i$th component, is recognized as a useful
instrument for the analysis of kinetic equations (especially in
chemical kinetics \cite{YBGE,Hangos2009}).

Each given positive distribution $Q$ could be represented as an
equilibrium Boltzmann distribution for given $T>0$ if we take
$u_i=-kT\log q_i + u_0$ for an arbitrary constant level $u_0$. If
we change the order of arguments in the Kullback--Leibler
divergence then we get the relative Burg entropy
\cite{Burg1967,Burg1972}. It has a much more exotic physical
interpretation: for a current distribution $P$ we can define the
``auxiliary energy" functional $U_P$ for which $P$ is the
equilibrium distribution under a given temperature $T$. We can
calculate the auxiliary free energy of any distribution $Q$ and
this auxiliary energy functional: $F_P(Q)$. (Up to an additive
constant, for $P=P^*$ this $F_P(Q)$ turns into the classical free
energy, $F_P^*(Q)=F(Q)$.) In particular, we can calculate the
auxiliary free energy of the physical equilibrium, $F_P(P^*)$. The
relative Burg entropy is
$$D_{\mathrm{KL}}(P^*\|P)=\frac{F_P(P^*)-F_P(P)}{kT}$$ This
functional should also decrease in any Markov process with given
equilibrium $P^*$.

Information theory developed by Shannon and his successors focused
on entropy as a measure of uncertainty of subjective choice.  This
understanding of entropy was returned from information theory to
statistical mechanics by E.T. Jaynes as a basis of ``subjective"
statistical mechanics \cite{Jaynes1957a,Jaynes1957b}. He followed
Wigner's idea ``entropy is an antropocentric concept". The entropy
maximum approach was declared as a minimization of the subjective
uncertainty. This approach gave rise to a MaxEnt ``anarchism". It
is based on a methodological hypothesis that everything unknown
could be estimated by the principle of the entropy maximum under
the condition of fixed known quantities. At this point the
classicism in entropy development changed to a sort of scientific
modernism. The art of model fitting based on entropy maximization
was developed \cite{Harre2001}. The principle of the entropy
maximum was applied to plenty of problems: from many physical
problems \cite{Beck2009}, chemical kinetics and process
engineering \cite{Hangos2009} to econometrics
\cite{MittJudge2000,Judge2002} and psychology \cite{Myers1992}.
Many new entropies were invented and now one has rich choice of
entropies for fitting needs \cite{EstMor1995}. The most celebrated
of them are the R\'enyi entropy \cite{Renyi1961}, the Burg entropy
\cite{Burg1967,Burg1972}, the Tsallis entropy \cite{Tsa1988, Abe}
and the Cressie--Read family \cite{CR1984,ReadCreass1988}. The
nonlinear generalized averaging operations and generalized entropy
maximization procedures were suggested \cite{Bagci2009}.

Following this impressive stream of works we understand the MaxEnt
approach as conditional maximization of entropy for the evaluation
of the probability distribution when our information is partial
and incomplete. The entropy function may be the classical BGS
entropy or any function from the rich family of non-classical
entropies. This rich choice causes a new problem: which entropy is
better for a given class of applications?

The MaxEnt ``anarchism" was criticized many times as a ``senseless
fitting". Arguments pro and contra the MaxEnt approach with
non-classical entropies (mostly the Tsallis entropy
\cite{Tsa1988}) were collected by Cho \cite{Cho2002}. This
sometimes ``messy and confusing situation regarding
entropy-related studies has provided opportunities for us: clearly
there are still many very interesting studies to pursue"
\cite{Lin1999}.

\subsection{Key Points}

In this paper we do not pretend to invent new entropies. (There
appear new functions as limiting cases of the known entropy
families, but this is not our main goal). Entropy is understood in
this paper as a measure of uncertainty which increases in Markov
processes. In our paper we consider a Markov process as a
semigroup on the space of positive probability distributions. The
state space is finite. Generalizations to compact state spaces are
simple. We analyze existent relative entropies (divergences) using
several simple ideas:
\begin{enumerate}
\item{In Markov processes
probability distributions $P(t)$ monotonically approach
equilibrium $P^*$: divergence $D(P(t)\|P^*)$ monotonically
decrease in time. }
\item{In most applications, conditional minimizers and maximizers
of entropies and divergences are used, but the values are not.
This means that the system of level sets is more important than
the functions' values. Hence, most of the important properties are
invariant with respect to monotonic transformations of entropy
scale.}
\item{The system of level sets should be the same as for additive
functions: after some rescaling the divergences of interest should
be additive with respect to the joining of statistically independent
systems. }
\item{The system of level sets should after
some rescaling the divergences of interest should have the form of
a sum (or integral) over states  $\sum_i f(p_i,p_i^*)$, where the
function $f$ is the same for all states. In information theory,
divergences of such form are called {\it separable}, in physics
the term {\it trace--form functions} is used}
\end{enumerate}

The first requirement means that if a distribution becomes  more
random then it becomes closer to equilibrium (Markov process
decreases the information excess over equilibrium). For example,
classical entropy increases in all Markov processes with uniform
equilibrium distributions. This is why we may say that the
distribution with higher entropy is more random, and why we use
entropy conditional maximization for the evaluation of the
probability distribution when our information is partial \linebreak
and incomplete.

It is worth to mention that some of the popular Bregman divergences,
for example, the squared Euclidean distance or the Itakura--Saito
divergence, do not satisfy the first requirement (see
Section~\ref{NoMore}).

The second idea is just a very general methodological thesis: to evaluate an instrument
(a divergence) we have to look how it works (produces conditional minimizers and
maximizers). The properties of the instrument which are not related to its work are not
important. The number three allows to separate variables if the system consists of
independent subsystems, the number four relates to separation of variables for partitions
of the space of probability distributions.

Amongst a rich world of relative entropies and divergences, only
two families meet these requirements. Both were proposed in 1984.
The Cressie--Read (CR) family \cite{CR1984,ReadCreass1988}:
 $$H_{{\rm CR}\ \lambda}(P \| P^*)=\frac{1}{\lambda (\lambda+1)}
 \sum_i p_i\left[ \left(\frac{p_i}{p_i^*} \right)^{\lambda} -1
 \right]\ , \ \ \lambda\in ]-\infty,\infty[$$
 and the convex combination of the Burg and Shannon relative
 entropies proposed in \cite{G11984} and further analyzed in \cite{ENTR1,ENTR2}:
$$H(P \| P^*)= \sum_i (\beta p_i- (1-\beta)p_i^* )\log \left(
\frac{p_i}{p_i^*}\right)\ , \ \ \beta \in [0,1] $$
 When $\lambda  \to 0$ the CR divergence tends to the KL divergence (the relative
 Shannon entropy) and when $\lambda  \to -1$ it turns into the
 Burg relative entropy. The Tsallis entropy was introduced four
 years later \cite{Tsa1988} and became very popular in
 thermostatistics (there are thousands of works that use or analyze this entropy
 \cite{TsallisBiblio2009}). The Tsallis entropy coincides (up to a
 constant multiplier $\lambda+1$) with the CR entropy for $\lambda
 \in ]-1,\infty[$ and there is no need to study it separately (see
 discussion in Section~\ref{LaevelSets}).

A new problem arose: which entropy is better for a specific
problem? Many authors compare performance of different entropies
and metrics  for various problems (see, for example,
\cite{Cachin1997,Davis2007}). In any case study it may be possible
to choose ``the best" entropy but in general we have no sufficient
reasons for such a choice. We propose a possible way to avoid the
choice of the best entropy.

Let us return to the idea: the distribution $Q$ is more random than
$P$ if there exists a continuous-time Markov process (with given
equilibrium distribution $P^*$) that transforms $P$ into $Q$. We say
in this case that $P$ and $Q$ are connected by the {\em Markov
preorder} with equilibrium $P^*$ and use notation $P \succ^0_{P^*}
Q$. The {\em Markov order} $\succ_{P^*}$ is the transitive closure
of the Markov preorder.

If a priori information gives us a set of possible distributions
$W$ then the conditionally ``maximally random distributions" (the
``distributions without additional information", the ``most
indefinite distributions" in $W$) should be the minimal elements
in $W$ with respect to Markov order. If a Markov process (with
equilibrium $P^*$) starts at such a minimal element $P$ then it
cannot produce any other distribution from $W$ because all
distributions which are more random that $P$ are situated outside
$W$. In this approach, the maximally random distributions under
given a priori information may be not unique. Such distributions
form a set which plays the same role as the standard MaxEnt
distribution. For the moment based a priori information the set
$W$ is an intersection of a linear manifold with the simplex of
probability distributions, the set of minimal elements in $W$ is
also polyhedron and its description is available in explicit form.
In low-dimensional case it is much simpler to construct this
polyhedron than to find the MaxEnt distributions for most of
specific entropies.

\subsection{Structure of the Paper}

The paper is organized as follows. In Section~\ref{NonClaEnt} we
describe the known non-classical divergences (relative entropies)
which are the Lyapunov functions for the Markov processes. We
discuss the general construction and the most popular families of
such functions. We pay special attention to the situations, when
different divergences define the same order on distributions and
provide the same solutions of the MaxEnt problems
(Section~\ref{LaevelSets}). In two short technical Sections
\ref{norma} and \ref{symm} we present normalization and
symmetrization of divergences (similar discussion of these
operations was published very recently \cite{Petz2010}.

The divergence between the current distribution and equilibrium
should decrease due to Markov processes. Moreover, divergence
between any two distributions should also decrease (the generalized
data processing Lemma,  Section~\ref{PRoductContract}).

Definition of entropy by its properties is discussed in
Section~\ref{DefPropert}. Various approaches to this definition
were developed for the BGS entropy by Shannon \cite{Shannon1948},
\cite{Renyi1970} and by other authors for the R\'enyi entropy
\cite{Aczel,AczDar}, the Tsallis entropy \cite{Abik4}, the CR
entropy and the convex combination of the BGS and Burg entropies
\cite{ENTR3}. Csisz\'ar \cite{Csiszar1978} axiomatically
characterized the class of Csisz\'ar--Morimoto divergences
(formula (\ref{Morimoto}) below). Another characterization of this
class was proved in \cite{ENTR3} (see Lemma 1,
Section~\ref{NoMore} below).

From the celebrated properties of entropy \cite{Wehrl} we selected
the following three:
\begin{enumerate}
\item{Entropy should be a Lyapunov function for continuous-time Markov
processes;}
\item{Entropy is additive with respect to the joining of independent
systems;}
\item{Entropy is additive with respect to the partitioning
of the space of states (\emph{i.e.}, has the {\em trace--form}).}
\end{enumerate}

To solve the MaxEnt problem we have to find the maximizers of
entropy (minimizers of the relative entropy) under given conditions.
For this purpose, we have to know the sublevel sets of entropy, but
not its values. We consider entropies with the same system of
sublevel sets as equivalent ones (Section~\ref{LaevelSets}). From
this point of view, all important properties have to be invariant
with respect to monotonic transformations of the entropy scale. Two
last properties from the list have to be substituted by the
following:
\begin{enumerate}
\item[2'.]{There exists a monotonic transformation which makes entropy additive with
respect to the joining of independent systems
(Section~\ref{additivity});}
\item[3'.]{There exists a monotonic transformation which makes entropy additive with
respect to the partitioning of the space of states
(Section~\ref{Partition}).}
\end{enumerate}
Several ``No More Entropies" Theorems are proven in Section~\ref{NoMore}: if an entropy
has properties 1, 2' and 3' then it belongs to one of the following one-parametric
families: to the Cressie--Read family, or to a convex combination of the classical BGS
entropy and the Burg entropy (may be, after a monotonic transformation of scale).

It seems very natural to consider divergences as orders on
distribution spaces (Section~\ref{FunctionOrOrder}), the sublevel
sets are the lower cones of this orders. For several functions,
$H_1(P), \ldots, H_k(P)$ the sets $\{Q\ | \ H_i(P)>H_i(Q) \ {\rm for
\ all } \ i\}$ give the simple generalization of the sublevel sets.
In Section~\ref{MarOd} we discuss the more general orders in which
continuous time Markov processes are monotone, define the Markov
order and fully characterize the local Markov order. The Markov
chains with detailed balance define the Markov order for general
Markov chains (Section~\ref{order}). It is surprising that there is
no necessity to consider other Markov chains for the order
characterization (Section~\ref{order}) because no reversibility is
assumed in this analysis.

In Section~\ref{OrderMaxEnt} we demonstrate how is it possible to
use the Markov order to reduce the uncertainty in the standard
settings when a priori information is given about values of some
moments. Approaches to construction of the most random
distributions are presented in Section~\ref{CondExtrOrder}.

Various approaches for the definition of the reference
distribution (or the generalized canonical distribution) are
compared in Section~\ref{GCD}.

In Conclusion we briefly discuss the main results.

\section{Non-Classical Entropies \label{NonClaEnt}}

\subsection{The Most Popular Divergences}
\subsubsection{Csisz\'ar--Morimoto Functions $H_h$}

During the time of modernism plenty of new entropies were
proposed. Esteban and Morales \cite{EstMor1995} attempted to
systemize many of them in an impressive table. Nevertheless, there
are relatively few entropies in use now. Most of the relative
entropies have the form proposed independently in 1963 by \linebreak I.
Csiszar \cite{Csiszar1963} and T. Morimoto \cite{Morimoto1963}:
\begin{equation}\label{Morimoto}
H_h(p)=H_h(P \| P^*)=\sum_i p^*_i h\left(\frac{p_i}{p_i^*}\right)
\end{equation}
where $h(x)$ is a convex function defined on the open ($x>0$) or
closed $x\geq 0$ semi-axis. We use here notation $H_h(P \| P^*)$
to stress the dependence of $H_h$ both on $p_i$ and $p^*_i$.

These relative entropies are the Lyapunov functions for all Markov
chains with equilibrium \linebreak $P^*=(p^*_i)$. Moreover, they
have the relative entropy contraction property given by the
generalized data processing lemma (Section~\ref{LyapDiver} below).

For $h(x)=x\log x$ this function coincides with the
Kullback--Leibler divergence from the current distribution $p_i$
to the equilibrium $p^*_i$. Some practically important functions
$h$ have singularities at 0. For example, if we take $h(x)=-\log
x$, then the correspondent $H_h$ is the relative Burg entropy
\linebreak $H_h=-\sum_i p^*_i \log (p_i/p_i^*) \to \infty$ for $p_i \to 0$.

\subsubsection{Required Properties of the Function $h(x)$}

Formally, $h(x)$ is an extended real-valued proper convex function
on the closed positive real half-line $[0,\infty[$. An {\em extended
real-valued function} can take real values and infinite values $\pm
\infty$. A {\em proper function} has at least one finite value. An
extended real valued function on a convex set $U$ is called {\em
convex} if \linebreak its {\em epigraph} $${\rm epi} (h) = \{(x,y)\
| \ x>0, y\geq h(x)\}$$ is a convex set \cite{Rockafellar1970}. For
a proper function this definition is equivalent to the {\em Jensen
inequality}

$$h(a x+(1-a)y)\leq ah(x)+(1-a) h(y) \;\; \mbox{for all}\;\;x,y\in
U, \, a\in[0,1]$$

It is assumed that the function $h(x)$ takes finite values on the
open positive real half-line $]0,\infty[$ but  the value at point
$x=0$ may be infinite. For example, functions $h(x)=- \log x$ or
$h(x)=x^{-\alpha}$ ($\alpha>0$) are allowed. A convex function
$h(x)$ with finite values on the open positive real half-line is
continuous on $]0,\infty[$ but may have a discontinuity at $x=0$.
For example, the step function, $h(x)=0$ if $x= 0$ and $h(x)=-1$
if $x> 0$, may be used.

A convex function is differentiable almost everywhere. Derivative
of $h(x)$ is a monotonic function which has left and right limits
at each point $x>0$. An inequality holds:  $h'(x)(y-x)\leq
h(y)-h(x)$ (Jensen's inequality in the differential form). It is
valid also for left and right limits of $h'$ at any point $x >0$.

Not everywhere differentiable functions $h(x)$ are often used, for
example, $h(x)=|x-1|$. Nevertheless, it is convenient to consider
the twice differentiable on $]0,\infty[$ functions $h(x)$ and to
produce a non-smooth $h(x)$ (if necessary) as a limit of smooth
convex functions. We use widely this possibility.

\subsubsection{The Most Popular Divergences $H_h(P \| P^*)$}

\begin{enumerate}
\item{Let $h(x)$ be the step function, $h(x)=0$ if $x=0$ and $h(x)=-1$ if $x>
0$. In this case,
\begin{equation}\label{Hartley}
H_h(P \| P^*)=-\sum_{i, \ p_i>0} 1
\end{equation}
The quantity  $-H_h$ is the number of non-zero probabilities $p_i$
and does not depend on $P^*$. Sometimes it is called the Hartley
entropy.}
\item{$h=|x-1|$, $$H_h(P \| P^*)=\sum_i |p_i - p_i ^*| $$
 this is the $l_1$-distance between $P$ and $P^*$.}
\item{$h=x\ln x$,
\begin{equation}\label{Kullback--Leibler}
H_h(P \| P^*)=\sum_i p_i \ln \left(\frac{p_i}{p_i^*}
\right)=D_{\mathrm{KL}}(P\|P^*)
\end{equation}
 this is the
usual Kullback--Leibler divergence or the relative BGS entropy;}
\item{$h=-\ln x$,
\begin{equation}\label{Burg}
H_h(P \| P^*)=-\sum_i p^*_i \ln \left(\frac{p_i}{p_i^*} \right) =
D_{\mathrm{KL}}(P^*\|P)
\end{equation}
this is the relative Burg entropy. It is obvious that this is again
the Kullback--Leibler divergence, but for another order of
arguments. }
\item{Convex combinations of $h=x\ln x$ and $h=-\ln x$ also produces a remarkable family
of divergences: $h=\beta x\ln x- (1-\beta) \ln x$ ($\beta \in
[0,1]$),
\begin{equation}\label{ConvComb}
H_h(P \| P^*)=\beta D_{\mathrm{KL}}(P\|P^*) +
(1-\beta)D_{\mathrm{KL}}(P^*\|P)
\end{equation}
 The convex combination of divergences becomes a symmetric functional of $(P, P^*)$ for
$\beta=1/2$. There exists a special name for this case,
``Jeffreys' entropy". }
\item{$h=\frac{(x-1)^2}{2}$,
\begin{equation}
H_h(P \| P^*)=\frac{1}{2}\sum_i \frac{(p_i-p_i^*)^2}{p_i^*}
\end{equation}
 This is the quadratic term in the Taylor expansion of
the relative Botzmann--Gibbs-Shannon entropy,
$D_{\mathrm{KL}}(P\|P^*)$, near equilibrium. Sometimes, this
quadratic form is called the Fisher entropy.}
\item{$h=\frac{x(x^{\lambda}-1)}{\lambda (\lambda+1)}$,
\begin{equation}\label{Cressie--Read}
H_h(P \| P^*)=\frac{1}{\lambda (\lambda+1)}
 \sum_i p_i\left[ \left(\frac{p_i}{p_i^*} \right)^{\lambda} -1 \right]
\end{equation}
  This is the  CR family
 of power divergences \cite{CR1984,ReadCreass1988}.
  For this family we use notation $H_{\rm CR \ \lambda}$.
   If $\lambda \to 0$ then $H_{\rm CR \ \lambda} \to  D_{\mathrm{KL}}(P\|P^*)$, this is the classical
 BGS relative entropy; if $\lambda \to -1$ then $H_{\rm CR \ \lambda} \to  D_{\mathrm{KL}}(P^*\|P)$,
 this is the relative Burg entropy. }
\item{For the CR family in the limits $\lambda \to
\pm \infty$ only the maximal terms ``survive". Exactly as we get the
limit $l^{\infty}$ of $l^p$ norms for $p \to \infty$, we can use the
root $({\lambda (\lambda+1)}H_{\rm CR \ \lambda})^{1/|\lambda|}$ for
$\lambda \to \pm \infty$ and write in these limits the divergences:
\begin{equation}\label{CR+infty}
 H_{\rm CR \ \infty}(P \| P^*) =\max_i\left\{\frac{p_i}{p_i^*}\right\}-1
\end{equation}
\begin{equation}\label{CR-infty}
H_{\rm CR \ -\infty} (P \| P^*)=
\max_i\left\{\frac{p_i^*}{p_i}\right\}-1 \end{equation}
 The existence of two limiting divergences $H_{{\rm CR \
\pm\infty}}$ seems very natural: there may be two types of extremely
non-equilibrium states: with a high excess of current probability
$p_i$ above $p_i^*$ and, inversely, with an extremely small current
probability $p_i$ with respect to $p_i^*$. }
\item{The Tsallis relative entropy \cite{Tsa1988} corresponds
to the choice $h=\frac{(x^{\alpha}-x)}{\alpha - 1}$, $\alpha >0$.
\begin{equation}\label{Tsallis}
H_h(P \| P^*)=\frac{1}{\alpha-1}\sum_i p_i\left[
\left(\frac{p_i}{p_i^*} \right)^{\alpha-1} -1 \right]
\end{equation}
For this family we use notation $H_{\rm Ts \
\alpha}$.}
\end{enumerate}

\subsubsection{R\'enyi Entropy}

The R\'enyi entropy of order $\alpha > 0 $ is defined
\cite{Renyi1961} as
\begin{equation}
H_{{\rm R}\ \alpha}(P) = \frac{1}{1-\alpha}\log\left(\sum_{i=1}^n
p_i^\alpha\right)
\end{equation}

It is a concave function, and $$H_{{\rm R}\ \alpha}(P) \to S(P)$$
when $\alpha \to 1$, where $S(P)$ is the Shannon entropy.

When $\alpha \to \infty$, the R\'enyi entropy has a limit
$H_\infty (X) = - \log \max_{i=1,\ldots n} p_i$, which has a
special name ``Min-entropy".

It is easy to get the expression for a relative R\'enyi entropy
$H_{{\rm R}\ \alpha}(P \| P^*)$ from the requirement that it
should be a Lyapunov function for any Markov chain with
equilibrium $P^*$:

$$H_{{\rm R}\ \alpha}(P \|
P^*)=\frac{1}{\alpha-1}\log\left(\sum_{i=1}^n p_i
\left(\frac{p_i}{p_i^*}\right)^{\alpha-1}\right)$$

For the Min-entropy, the correspondent divergence (the relative
Min-entropy) is

$$H_\infty (P \| P^*) = \log \max_{i=1,\ldots n}
\left(\frac{p_i}{p_i^*}\right)$$
 It is obvious from (\ref{MAsterEq1}) below that $\max_{i=1,\ldots n}
({p_i}/{p_i^*})$ is a Lyapunov function for any Markov chain with
equilibrium $P^*$, hence, the relative Min-entropy is also the
Lyapunov functional.

\subsection{Entropy Level Sets \label{LaevelSets}}

A {\em level set} of a real-valued function $f$ is a set of the form : $$\{x \ | \ f(x)=c
\}$$ where $c$ is a constant (the ``level"). It is the set where the function takes on a
given constant value. A {\em sublevel set} of $f$ is a set of the form
$$\{x \ | \ f(x)\leq c \}$$ A {\em superlevel set } of $f$ is
given by the inequality with reverse sign: $$\{x \ | \ f(x)\geq c
\}$$ The intersection of the sublevel and the superlevel sets for
a given value $c$ is the level set for this level.

In many applications, we do not need the entropy values, but rather
the order of these values on the line. For any two distributions
$P,Q$ we have to compare which one is closer to equilibrium $P^*$,
\emph{i.e.}, to answer the question: which of the following
relations is true: $H(P\|P^*)>H(Q\|P^*)$, $H(P\|P^*)=H(Q\|P^*)$ or
$H(P\|P^*)<H(Q\|P^*)$? To solve the MaxEnt problem we have to find
the maximizers of entropy (or, in more general settings, the
minimizers of the relative entropy) under given conditions. For this
purpose, we have to know the sublevel sets, but not the values. All
the MaxEnt approach does not need the values of the entropy but the
sublevel sets are necessary.

Let us consider two functions, $\phi$ and $\psi$ on a set $U$. For
any $V \subset U$ we can study conditional minimization problems
$\phi (x) \to \min$ and $\psi (x) \to \min$, $x \in V$. The sets of
minimizers for these two problems coincide for any $V \subset U$ if
and only if the functions $\phi$ and $\psi$ have the same sets of
sublevel sets. It should be stressed that here just the sets of
sublevel sets have to coincide without any relation to values of
level.

Let us compare the level sets for the R\'enyi, the Cressie-Read
and the Tsallis families of divergences (for $\alpha-1=\lambda$
and for all values of $\alpha$). The values of these functions are
different, but the level sets are the same (outside the Burg
limit, where $\alpha \to 0$): for $\alpha \neq 0, 1$

\begin{equation}
H_{{\rm R}\ \alpha}(P \| P^*)=\frac{1}{\alpha-1}\ln c;  \ \
H_{\rm CR \ \alpha-1}(P \| P^*)= \frac{1}{\alpha(\alpha-1)}(c-1);
\ \
 H_{\rm Ts \ \alpha}(P \| P^*) =\frac{1}{\alpha-1}(c-1)
\end{equation}
 where $c=\sum_i p_i (p_i/p_i^*)^{\alpha-1}$.

Beyond points $\alpha=0,1$
 $$H_{\rm CR \ \alpha-1}(P \| P^*)=\frac{1}{\alpha(\alpha-1)} \exp
 ((\alpha-1)H_{{\rm R}\ \alpha}(P \| P^*))=\frac{1}{\alpha}H_{\rm Ts \ \alpha}(P \|
 P^*)$$

For $\alpha \to 1$ all these divergences turn into the Shannon
relative entropy. Hence, if $\alpha \neq 0$ then for any $P$,
$P^*$, $Q$, $Q^*$ the following equalities A, B, C are equivalent,
A$\Leftrightarrow$B$\Leftrightarrow$C:
\begin{equation}
\begin{split}
&{\rm A}.\;\; H_{{\rm R}\ \alpha}(P \| P^*)=H_{{\rm R}\ \alpha}(Q
\| Q^*) \\
 &{\rm B}.\;\; H_{{\rm  CR}\ \alpha+1}(P \| P^*)=H_{{\rm CR}\
\alpha+1}(Q \| Q^*) \\
 &{\rm C}.\;\;H_{{\rm Ts}\ \alpha}(P \|
P^*)=H_{{\rm Ts}\ \alpha}(Q \| Q^*)
\end{split}
\end{equation}

This equivalence means that we can select any of these three
divergences as a basic function and consider the others as
functions of this basic one.

For any $\alpha \geq 0$ and $\lambda=\alpha+1$ the R\'enyi, the
Cressie--Read and the Tsallis divergences have the same family of
sublevel sets. Hence, they give the same maximizers to the
conditional relative entropy minimization problems and there is no
difference which entropy to use.

The CR family has a more convenient normalization factor
$1/\lambda (\lambda+1)$ which gives a proper convexity for all
powers, both positive and negative, and provides a sensible Burg
limit for $\lambda \to -1$ (in contrary, when  $\alpha \to 0$ both
the R\'enyi and Tsallis entropies tend to 0).

When $\alpha < 0$ then for the Tsallis entropy function
$h=\frac{(x^{\alpha}-x)}{\alpha - 1}$ loses convexity, whereas for
the Cressie-Read family convexity persists for all values of
$\lambda$. The R\'enyi entropy also loses convexity for $\alpha <
0$. Neither the Tsallis, nor the R\'enyi entropy were invented for
use with negative $\alpha$.

There may be a reason: for $\alpha < 0$ the function $x^{\alpha}$ is
defined for $x>0$ only and has a singularity at $x=0$. If we assume
that the divergence should exist for all non-negative distributions,
then the cases $\alpha \leq 0$ should be excluded. Nevertheless, the
Burg entropy which is singular at zeros is practically important and
has various attractive properties. The Jeffreys' entropy (the
symmetrized Kullback--Leibler divergence) is also singular at zero,
but has many important properties. We can conclude at this point
that it is not obvious that we have to exclude any singularity at
zero probability. It may be useful to consider positive
probabilities instead (``nature abhors a vacuum").

Finally, for the MaxEnt approach (conditional minimization of the
relative entropy), the R\'enyi and the Tsallis families of
divergences ($\alpha > 0$) are particular cases of the Cressie--Read
family because they give the same minimizers. For $\alpha \leq 0$
the R\'enyi and the Tsallis relative entropies lose their convexity,
while the Cressie--Read family remains convex for $\lambda \leq -1$
too.

\subsection{Minima and normalization \label{norma}}

For a given $P^*$, the function $H_h$ achieves its minimum on the
hyperplane $\sum_i p_i =\sum_i p_i^*=$const at equilibrium
$p_i^*$, because at this point
 $${\rm grad} H_h =(h'(1),\ldots h'(1))=h'(1){\rm grad}\left(\sum_i p_i\right)$$
The transformation $h(x) \to h(x) + ax +b$ just shifts $H_h$ by
constant value: $H_h \to H_h + a\sum_i p_i + b=H_h+a+b$.
Therefore, we can always assume that the function $h(x)$ achieves
its minimal value at point $x=1$, and this value is zero. For this
purpose, one should just transform $h$:

\begin{equation}\label{normalization}
h(x):=h(x)-h(1)-h'(1)(x-1)
\end{equation}

This normalization transforms $x \ln x$ into $x \ln x - (x-1)$, $-
\ln x$ into $-\ln x + (x-1)$, and $x^{\alpha}$ into $x^{\alpha} -
1 - \alpha (x-1)$. After normalization $H_h(P \| P^*) \geq 0$. If
the normalized $h(x)$ is strictly positive outside point $x=1$
($h(x)> 0$ if $x \neq 1$) then $H_h(P \| P^*) = 0$ if and only if
$P = P^*$ (\emph{i.e.}, in equilibrium).

The normalized version of any divergence $H_h(P \| P^*)$ could be
produced by the normalization transformation
$h(x):=h(x)-h(1)-h'(1)(x-1)$ and does not need separate discussion.

\subsection{Symmetrization \label{symm}}

Another technical issue is symmetry of a divergence. If $h(x)=x\ln
x$ then both $H_h(P\| P^*)$ (the KL divergence) and $H_h(P^* \| P)$
(the relative Burg entropy) are the Lyapunov functions for the
Markov chains, and $H_h(P^* \| P)=H_g(P\| P^*)$ with $g(x)=-\ln x$.
Analogously, for any $h(x)$ we can write $H_h(P^*\| P)=H_g(P\| P^*)$
with
\begin{equation}\label{involution}
g(x)=xh\left(\frac{1}{x}\right)
\end{equation}
If $h(x)$ is convex on $\mathbf{R}_+$ then $g(x)$ is convex on
$\mathbf{R}_+$ too because
$$g''(x)=\frac{1}{x^3}h''\left(\frac{1}{x}\right) $$ The
transformation (\ref{involution}) is an involution:
$$xg\left(\frac{1}{x}\right)=h(x) $$

The fixed points of this involution are such functions  $h(x)$
that $H_h(P\| P^*)$ is symmetric with respect to transpositions of
$P$ and $P^*$. There are many such $h(x)$. An example of symmetric
$H_h(P\| P^*)$ gives the choice $h(x)=-\sqrt{x}$: $$H_h(P\|
P^*)=-\sum_i \sqrt{p_i p^*_i} $$ After normalization, we get
 $$h(x):=\frac{1}{2}(\sqrt{x}- 1)^2 \ ;
\; \; \; H_h(P\| P^*)=\frac{1}{2}\sum_i (\sqrt{p_i}-
\sqrt{p^*_i})^2 $$
 Essentially (up to a constant addition and multiplier) this
function coincides with a member of the CR family, $H_{{\rm CR } \
-\frac{1}{2}}$ (\ref{Cressie--Read}), and with one of the Tsallis
relative entropies  $H_{{\rm Ts } \ \frac{1}{2}}$ (\ref{Tsallis}).
The involution (\ref{involution}) is a linear operator, hence, for
any convex $h(x)$ we can produce its symmetrization:

$$h_{\rm sym}(x)=\frac{1}{2}(h(x)+g(x))=
\frac{1}{2}\left(h(x)+xh\left(\frac{1}{x}\right)\right) $$ For
example, if $h(x)=x \log x$ then $h_{\rm sym}(x)=\frac{1}{2}(x
\log x - \log x)$; if $h(x) = x^{\alpha}$ then $h_{\rm
sym}(x)=\frac{1}{2}(x^{\alpha} + x^{1-\alpha})$.

\section{Entropy Production and Relative Entropy Contraction \label{PRoductContract}}

\subsection{Lyapunov Functionals for Markov Chains}

Let us consider continuous time Markov chains with positive
equilibrium probabilities $p_j^*$. The dynamics of the
probability distribution $p_i$ satisfy the Master equation (the
Kolmogorov equation):
\begin{equation}\label{MAsterEq0}
\frac{\D p_i}{\D t}= \sum_{j, \, j\neq i} (q_{ij}p_j-q_{ji}p_i)
\end{equation}
 where coefficients $q_{ij}$ ($i\neq j$) are non-negative. For
 chains with a positive equilibrium distribution $p_j^*$ another
 equivalent form is convenient:
\begin{equation}\label{MAsterEq1}
\frac{\D p_i}{\D t}= \sum_{j, \, j\neq i}
q_{ij}p^*_j\left(\frac{p_j}{p_j^*}-\frac{p_i}{p_i^*}\right) \
\end{equation}
where $p_i^*$ and $q_{ij}$ are connected by identity
\begin{equation}\label{MasterEquilibrium}
\sum_{j, \, j\neq i} q_{ij}p^*_j = \left(\sum_{j, \, j\neq i}
q_{ji}\right)p^*_i
\end{equation}

The time derivative of the Csisz\'ar--Morimoto function $H_h(p)$
(\ref{Morimoto}) due to the Master equation is
\begin{equation}\label{ENtropyProd}
\frac{\D H_h(P \| P^*)}{\D t}=\sum_{i,j, \, j\neq i}
q_{ij}p^*_j\left[h\left(\frac{p_i}{p_i^*}\right)-h\left(\frac{p_j}{p_j^*}\right)
+ h'\left(\frac{p_i}{p_i^*}\right)\left(\frac{p_j}{p_j^*}-
\frac{p_i}{p_i^*}\right)\right] \leq 0
\end{equation}
To prove this formula, it is worth to mention that for any $n$
numbers $h_i$, $\sum_{i,j, \, j\neq i} q_{ij}p^*_j
(h_i-h_j)=0$. The last inequality holds because of the
convexity of $h(x)$:
 $h'(x)(y-x)\leq h(y)-h(x)$ \linebreak (Jensen's inequality).

Inversely, if
\begin{equation}\label{ENtropyProd2}
h(x)-h(y) + h'(y)(x- y) \leq 0
\end{equation}
for all positive $x,y$ then $h(x)$ is convex on $\mathbf{R}_+$.
Therefore, if for some function $h(x)$  $H_h(p)$ is the Lyapunov
function for all the Markov chains with equilibrium $P^*$ then
$h(x)$ is convex on $\mathbf{R}_+$.

The Lyapunov functionals $H_h$ do not depend on the kinetic
coefficients $q_{ij}$ directly. They depend on the equilibrium
distribution $p^*$ which satisfies the identity
(\ref{MasterEquilibrium}). This independence of the kinetic
coefficients is the {\em universality} property.

\subsection{``Lyapunov Divergences" for Discrete Time Markov
Chains \label{LyapDiver}}

The Csisz\'ar--Morimoto functions (\ref{Morimoto}) are also
Lyapunov functions for discrete time Markov chains. Moreover, they
can serve as a ``Lyapunov distances" \cite{Liese1987} between
distributions which decreases due to time evolution (and not only
the divergence between the current distribution and equilibrium).
In more detail, let $A=(a_{ij})$ be a stochastic matrix in
columns: $$a_{ij} \geq 0, \; \; \sum_i a_{ij}=1 \; {\rm for \;
all} \; j $$
 The {\em ergodicity contraction coefficient} for $A$ is a number $\overline{\alpha}(A)$
 \cite{DobrushinErgCoeff1956,Seneta1981}:
 $$\overline{\alpha}(A)=\frac{1}{2}\max_{i,k} \left\{\sum_j |a_{ij}-a_{kj}| \right\} $$
 $0 \leq \overline{\alpha}(A) \leq 1$.

Let us consider in this subsection the normalized
Csisz\'ar--Morimoto divergences $H_h(P \| Q)$
(\ref{normalization}): $h(1)=0, \, h(x)\geq 0$.

\vspace{2mm}\noindent{\bf Theorem about relative entropy
contraction.} {\it (The generalized data processing Lemma.) For each
two probability positive distributions $P,Q$ the divergence $H_h(P
\| Q)$ decreases under action of stochastic matrix $A$
\cite{Cohen1993,CohenIwasa1993}:
\begin{equation}\label{contractionH}
H_h(AP \| AQ)\leq \overline{\alpha}(A) H_h(P \| Q)
\end{equation}\vspace{2mm} }

The generalizations of this theorem for general Markov kernels
seen as operators on spaces of probability measures was given by
\cite{Ledoux2003}. The shift in time for continuous-time Markov
chain is a column-stochastic matrix, hence, this contraction
theorem is also valid for continuous-time \linebreak Markov chains.

The question about a converse theorem arises immediately. Let the
contraction inequality hold for two pairs of positive
distributions $(P,Q)$ and $(U,V)$ and for all $H_h$:
\begin{equation}\label{contraction22}
H_h(U \| V) \leq  H_h(P \| Q)
\end{equation}
Could we expect that there exists such a stochastic matrix $A$
that $U=AP$ and $V=AP$? The answer \linebreak is positive:

\vspace{2mm}\noindent{\bf The converse generalized data processing
lemma.} {\it Let the contraction inequality (\ref{contraction22})
hold for two pairs of positive distributions $(P,Q)$ and $(U,V)$
and for all normalized $H_h$. Then there exists such a
column-stochastic matrix $A$ that $U=AP$ and $V=AQ$
\cite{Cohen1993}.\vspace{2mm}}

This means that for the system of inequalities
(\ref{contraction22}) (for all normalized convex functions $h$ on
$]0,\infty[$) is necessary and sufficient for existence of a
(discrete time) Markov process which transform the pair of
positive distributions $(P,Q)$ in $(U,V)$. It is easy to show that
for continuous-time Markov chains this theorem is not valid: the
attainable regions for them are strictly smaller than the set
given by \linebreak inequalities (\ref{contraction22}) and could
be even non-convex (see \cite{Gorban1979} and
Section~\ref{SecHistory} below).

\section{Definition of Entropy by its Properties
\label{DefPropert}}

\subsection{Separation of Variables for Partition of the State Space
\label{Partition}}

An important property of separation of variables is valid for all divergences which have
the form of a sum of  convex functions $f(p_i,p_i^*)$. Let the set of states be divided
into two subsets, $I_1$ and $I_2$, and let the functionals $u^1,\ldots u^m$ be linear. We
represent each probability distribution as a direct sum $P=P^1\oplus P^2$, where
$P^{1,2}$ are restrictions of $P$ on $I_{1,2}$.

Let us consider the problem $$H(P\|P^*)\to \min$$ subject to
conditions $u^i(P)=U_i$ for a set of linear functionals $u^i(P)$.

The solution $P^{\min}$ to this problem has a form
$P^{\min}=P_1^{\min}\oplus P_2^{\min}$, where $P^{1,2}$ are
solutions to \linebreak the problems $$H(P^{1,2}\|P^{* \ 1,2})\to \min$$
subject to conditions $u^i(P_{1,2})=U_i^{1,2}$ and $\sum_{i \in
I_{1,2} }p_i^{1,2}=\pi_{1,2}$ for some redistribution of the
linear functionals values,  $U_i=U_i^1+U_i^2$, and of the total
probability, $1=\pi_1+\pi_2$ ($\pi_{1,2} \geq 0$) .

The solution to the divergence minimization problem is composed from solutions of the
partial maximization problems. Let us call this property the {\it separation of variables
for incompatible events} (because $I_1 \cap I_2 =\emptyset$).

This property is trivially valid for the Tsallis family (for
$\alpha>0$, and for $\alpha <0$ with the change of minimization to
maximization) and for the CR family. For the R\'enyi family it
also holds (for $\alpha>0$, and for $\alpha <0$ with the change
from minimization to maximization), because the R\'enyi entropy is
a function of those trace--form  entropies, their level sets
coincide.

A simple check shows that this separation of variables property holds also for the convex
combination of Shannon's and Burg's entropies, $\beta D_{\mathrm{KL}}(P\|P^*) +
(1-\beta)D_{\mathrm{KL}}(P^*\|P)$.

\subsection{Additivity Property \label{additivity}}

The {\em additivity property } with respect to joining of subsystems is crucial both for
the classical thermodynamics and for the information theory.

Let us consider a system which is result of joining of two subsystems. A state of the
system is an ordered pair of the states of the subsystems and the space of states of the
system is the Cartesian product of the subsystems spaces of state. For systems with
finite number of states this means that if the states of subsystems are enumerated by
indexes $j$ and $k$ then the states of the system are enumerated by pairs $jk$. The
probability distribution for the whole system is $p_{jk}$, and for the subsystems the
probability distributions are the marginal distributions $q_j=\sum_k p_{jk}$, $r_k=\sum_j
p_{jk}$.

The {\em additive functions of state} are defined for each  state of the subsystems and
for a state of the whole system they are sums of these subsystem values:
$$u_{jk}=v_j+w_k$$ where $v_j$ and $w_k$ are functions of the subsystems state.

In classical thermodynamics such functions are called the {\em extensive quantities}. For
expected values of additive quantities the similar additivity condition holds:
\begin{equation}\label{AdditivityCond}
\sum_{j,k} u_{jk} p_{jk}=\sum_{j,k} (v_j + w_k) p_{ik}=\sum_j v_j
q_j + \sum_k w_k r_k
\end{equation}
 Let us consider these expected values as functionals of the
 probability distributions: $u(\{p_{jk}\})$, $v(\{q_j\})$
 and $w(\{r_k\})$. Then the additivity property for the expected
 values reads:
\begin{equation}\label{additivefunctionals}
 u(\{p_{jk}\})=v(\{q_j\})+w(\{r_k\})
\end{equation}
where $q_j$ and the $r_k$ are the marginal distributions.

Such a linear additivity property is impossible for non-linear entropy functionals, but
under some independence conditions the entropy can behave as an extensive variable.

Let $P$ be a product of marginal distributions. This means that the subsystems are
statistically independent: $p_{jk}=q_j r_k$. Assume also that the distribution $P^*$ is
also a product of marginal distributions $p^*_{jk}=q^*_j r^*_k$. Then some entropies
reveal the additivity property with respect to joining of \linebreak independent systems.
\begin{enumerate}
\item{The BGS relative entropy $D_{\mathrm{KL}}(P\|P^*)=
    D_{\mathrm{KL}}(Q\|Q^*)+D_{\mathrm{KL}}(R\|R^*)$.}
\item{The Burg entropy $D_{\mathrm{KL}}(P^*\|P)=D_{\mathrm{KL}}(Q^*\|Q)+
    D_{\mathrm{KL}}(R^*\|R)$ . It is obvious that a convex combination of the Shannon
    and Burg entropies has the same additivity property.}
\item{The R\'enyi entropy $H_{{\rm R}\ \alpha}(P\| P^*)=H_{{\rm R}\ \alpha}(Q \|
    Q^*)+H_{{\rm R}\ \alpha}(R \| R^*) $. For $\alpha \to \infty$ the Min-entropy
    also inherits this property.}
\end{enumerate}

This property implies the separation of variables for the entropy maximization problems
if the system consists of independent subsystems, $p_{jk}=q_jr_k$. Let functionals
$u^1(\{p_{jk}\}), \ldots u^m(\{p_{jk}\})$ be additive (\ref{AdditivityCond})
(\ref{additivefunctionals}) and let the relative entropy $H(P\| P^*)$ be additive with
respect to joining of independent systems. Assume that in equilibrium subsystems are also
independent, $p^*_{jk}=q^*_jr^*_k$. Then the solution to the problem
 $$H(P\| P^*) \to \min$$ subject to conditions
\begin{equation}\label{condIndep}
 u^i (P)=U_i \ \ (i=1,\ldots m);\ \ p_{jk}=q_jr_k
 \end{equation}
 is $p_{jk}^{\min}=q_j^{\min}r_k^{\min}$, where $q_j^{\min}$,
 $r_k^{\min}$ are solutions of partial problems:
 $$H(Q\| Q^*) \to \min$$ subject to the conditions
 $$v^i (Q)=V_i \ \ (i=1,\ldots m) $$
 and $$H(R\| R^*) \to \min$$ subject to the conditions
 $$w^i (Q)=W_i \ \ (i=1,\ldots m)$$
 for some redistribution of the additive functionals values
 $U_i=V_i+W_i$.

Let us call this property the {\it separation of variables for independent subsystems}.

Neither the CR, nor the Tsallis divergences families have the additivity property. It is
proven \cite{ENTR3} that a function $H_h$ has the additivity property if and only if it
is a convex combination of the Shannon and Burg entropies. See also
Theorem~\ref{thAddTrace} in Appendix.

Nevertheless, both the CR and the Tsallis families have the property of separation of
variables for independent subsystems because of the coincidence of the level sets with
the additive function, the R\'enyi entropy (for all $\alpha> 0$).

The Tsallis entropy family has absolutely the same property of separation of variables as
the R\'enyi entropy. To extend this property of the R\'enyi Tsallis entropies for
negative $\alpha$, we have to change there min to max.

For the CR family the result sounds even better: because of better normalization, the
separation of variables is valid for $H_{\rm CR \ \lambda} \to \min$ problem for all
values $\lambda\in ]-\infty, \infty[$.

The condition of independence of subsystems $p_{jk}=q_jr_k$ in (\ref{condIndep}) cannot
be relaxed: if we assume  $p^*_{jk}=q^*_jr^*_k$ only then the correlations between
subsystems may emerge in the solution of the minimization problem. For example, without
assumption of independence, for the Burg entropy, the method of Lagrange multipliers
gives ($\phi_i$ and $\psi_i$ are the Lagrange multipliers):
$$\frac{p^*_{jk}}{p_{jk}^{\min}}=\sum_i (\phi_i v_j^i + \psi_i w_k^i)$$
and the subsystems are not independent in this state even if they are independent in
equilibrium and the conditions are additive. These emergent correlations may be
considered as spurious \cite{PresseGhosh2013} or may be interpreted as sensible ones for
some finite systems far from thermodynamic limit for modelling of non-canonic ensembles
\cite{ENTR1}. In any case, the use of entropies which are additive with respect to
joining of independent subsystem does not guarantee independence of subsystems but allows
only to separate variables under condition of independence.

The stronger condition was used by Shore and Johnson \cite{Shore1980} in the axiomatic
derivation of the principle of maximum entropy and the principle of minimum divergence
(or `cross-entropy'). They postulated that the MaxEnt distribution for the whole system
is the product of the distributions of the subsystems if the known information
(conditions) is the information about subsystems (Axiom III). Independence of subsystems
in this axioms is not assumed but should be the consequence of the entropy maximization.
This axiom can be called `separation of variables under independent conditions'. They
supplement this assumption by the separation of variables for partition of the state
space (Axiom IV), by the condition of uniqueness of the MaxEnt distribution (Axiom I),
and by the requirement of the invariance with respect to the coordinate transformations
(Axiom II). All these axioms together give the unique classical BGS entropy. For further
discussion see \cite{PresseGhosh2013}.

Violation of the Shore and Johnson Axiom III leads to correlation between subsystems and
this  is an essential difference of the non-classical MaxEnt ensembles from the classical
canonical ensembles.

We use the weaker assumption of separation of variables for {\em independent subsystems
and additive conditions}. Its violation leads to much more counterintuitive consequences:
Subsystems remain independent (condition) and other conditions are additive
(\ref{condIndep}) but the solution of the MaxEnt problem is the product of distributions
which are not solutions of the partial MaxEnt problems. In other words, the probability
distribution for a subsystem is modified just by existence of another subsystem without
any interactions and correlations.

It seems to be difficult to find a reason for such a behavior and therefore the
assumption of separation of variables for independent subsystems and additive conditions
is a sensible axiom. It is weaker than the Shore and Johnson Axiom III \cite{Shore1980}
and, therefore, leads to a wider family of entropies than just a classical BGS entropy.
This wider family includes the CR family (\ref{Cressie--Read}) and the convex combination
of the Shannon and the Burg entropies (\ref{ConvComb}).

The question arises: is there any new divergence that has the following three properties:
(i) the divergence $H(P\|P^*)$ should decrease in Markov processes with equilibrium
$P^*$, (ii) for minimization problems the separation of variables for independent
subsystems holds and (iii) the separation of variables for incompatible events holds. A
{\it new} divergence means here that it is not a function of a divergence from the CR
family or from the convex combination of the Shannon and the Burg entropies.

The answer is: no, any divergence which has these three properties and is defined and
differentiable for positive distributions is a monotone function of $H_h$ for
$h(x)={\alpha}p^{\alpha}$ (${\alpha}\in ]-\infty,\infty[$, ${\alpha}\neq 0,1 $), that is,
essentially,  the CR family (\ref{Cressie--Read}), or $h(x)=\beta x \ln x - (1-\beta) \ln
x$ ($\beta\in [0,1]$). If we relax the differentiability property, then we have to add to
the CR family the limits for  $\lambda \to \pm \infty$. For  $\lambda \to + \infty$ we
get the CR analogue of min-entropy
$$H_{\rm CR \ \infty} (P\|P^*)=
\max_i\left\{\frac{p_i}{p_i^*}\right\}-1 $$

The limiting case for the CR family for $\lambda \to - \infty$ is less known but is also
a continuous and piecewise differential Lyapunov function for the Master equation:

$$H_{{\rm CR \ -\infty}} (P\|P^*)=
\max_i\left\{\frac{p_i^*}{p_i}\right\}-1 $$

Both properties of separation of variables are based on the specific additivity
properties: additivity with respect to the composition of independent systems and
additivity with respect to the partitioning of the space of states. Separation of
variables can be considered as a weakened form of additivity: not the minimized function
should be additive but there exists such a monotonic transformation of scale after which
the function becomes additive (and different transformations may be needed for different
additivity properties).

\subsection{``No More Entropies" Theorems \label{NoMore}}

The classical Shannon work included the characterization of
entropy by its properties. This meant that the classical notion of
entropy is natural, and {\it no more entropies} are expected. In
the seminal work of R\'enyi, again the characterization of entropy
by its properties was proved, and for this, extended family the
{\it no more entropies} theorem was proved too. In this section,
we prove the next {\it no more entropies} theorem, where two
one-parametric families are selected as sensible: the CR family
and the convex combination of Shannon's and Burg's entropies. They
are two branches of solutions of the correspondent functional
equation and intersect at two points: Shannon's entropy
($\lambda=1$ in the CR family) and Burg's entropy ($\lambda=0$).
We consider entropies as equivalent if their level sets coincide.
In that sense, the R\'enyi entropy and the Tsallis entropy (with
$\alpha>0$) are equivalent to the CR entropy with $\alpha-1 =
\lambda$, $\lambda>-1$.

Following R\'enyi, we consider entropies of {\it incomplete
distributions}: $p_i \geq 0$, $\sum_i p_i \leq 1$. The divergence
$H(P\| P^*)$ is  a $C^1$ smooth function of a pair of positive
generalized probability distributions $P=(p_i)$, $p_i>0$ and
$P^*=(p^*_i)$, $p^*_i>0$, $i=1, \ldots n$.

The following 3 properties are required for characterization of
the ``natural" entropies.

\begin{enumerate}
\item{To provide the separation of variables for incompatible events together with
    the symmetry property we assume that the divergence is separable, possibly, after
    a scaling transformation: there exists such a function of two variables
    $f(p,p^*)$ and a monotonic function of one variable $\phi(x)$ that
    $H(P\|P^*)=\phi(\sum_i f(p_i,p^*_i))$. This formula allows us to define
    $H(P\|P^*)$ for all $n$.}
 \item{$H(P\| P^*)$ is a Lyapunov function for the
     Kolmogorov equation (\ref{MAsterEq1}) for any Markov
     chain with equilibrium $P^*$. (One can call these
     functions the {\it universal} Lyapunov functions
     because they do not depend on the   kinetic
     coefficients directly, but only on the equilibrium
     distribution $P^*$.)}
 \item{To provide separation of variables for independent
subsystems we assume that $H(P\| P^*)$ is additive (possibly after a
scaling transformation): there exists such a function of one
variable $\psi(x)$ that the function $\psi(H(P\|P^*))$ is additive
for the union of independent subsystems: if $P=(p_{ij})$,
$p_{ij}=q_j r_j$, $p^*_{ij}=q^*_j r^*_j$, then
$\psi(H(P\|P^*))=\psi(H(Q\|Q^*))+\psi(H(R\|R^*))$.}
\end{enumerate}

\begin{theorem}\label{Theorem1NoMore} If a $C^1$-smooth divergence $H(P\| P^*)$
satisfies the conditions 1-3 then, up to monotonic transformation,
it is either the CR divergence $H_{\rm CR \ \lambda}$ or a convex
combination of the Botlzmann--Gibbs--Shannon and the Burg
entropies, $H_h(P \| P^*)=\beta D_{\mathrm{KL}}(P\|P^*) +
(1-\beta)D_{\mathrm{KL}}(P^*\|P)$.
\end{theorem}

In a paper \cite{ENTR3} this family was identified as the Tsallis
relative entropy with some abuse of language, because in the
Tsallis entropy the case with $\alpha<0$ is usually excluded.

First of all, let us prove that any function which satisfies the
conditions 1 and 2 is a monotone function of a Csisz\'ar--Morimoto
function (\ref{Morimoto}) for some convex smooth function $h(x)$.
This was mentioned in 2003 by P. Gorban \cite{ENTR3}. Recently, a
similar statement was published by S. Amari (Theorem 1 in
\cite{Amari2009}).

\begin{lemma}\label{MorimotoChar}If a Lyapunov function $H(p)$ for the Markov chain is of the
trace--form ($H(p)=\sum_{i}f(p_{i},p_{i}^{*})$) and is universal,
then $f(p,p^{*})=p^{*}h(\frac{p}{p^{*}})+{\rm const}(p^*)$, where
$h(x)$ is a convex function of one variable.
\end{lemma}

\begin{proof}
 Let us consider a Markov chain with two states. For such a chain
\begin{equation}
\frac{\D p_1}{\D t}=q_{12}p_{2}^{*}\left(\frac{p_{2}}{p_{2}^{*}}-
\frac{p_{1}}{p_{1}^{*}}\right)=
-q_{21}p_{1}^{*}\left(\frac{p_{1}}{p_{1}^{*}}-\frac{p_{2}}{p_{2}^{*}}\right)=-\frac{\D
p_{2}}{\D t}
\end{equation}
If $H$ is a Lyapunov function then $\dot{H}\leq 0$ and the following
inequality holds:

$$\left(\frac{\partial f(p_{2},p_{2}^{*})}{\partial p_{2}}-\frac{\partial
f(p_{1},p_{1}^{*})}{\partial
p_{1}}\right)\left(\frac{p_{1}}{p_{1}^{*}}-\frac{p_{2}}{p_{2}^{*}}\right)\leq0
$$

We can consider $p_1,p_2$ as independent variables from an open
triangle $D=\{(p_1,p_2)\ | \ p_{1,2} > 0, \ p_1+p_2 < 1 \}$. For
this purpose, we can include the Markov with two states into a
chain with three states and $q_{3i}=q_{i3}=0$.

If for a continuous function of two variables $\psi (x,y)$ in an
open domain $D \subset \mathbb{R}^2$ an inequality $(\psi
(x_1,y_1) - \psi (x_2,y_2))(y_1-y_2) \leq 0$ holds then this
function does not depend on $x$ in $D$. Indeed, let there exist
such values $x_{1,2}$ and $y$ that $\psi(x_1,y) \neq \psi(x_2,y)$,
$\psi(x_1,y) - \psi(x_2,y)=\varepsilon >0$. We can find such
$\delta > 0$ that $(x_{1},y+\Delta y) \in D$ and
$|\psi(x_{1},y+\Delta y) - \psi(x_{1},y)|<\varepsilon/2$ if
$|\Delta y|< \delta$. Hence, $\psi(x_1,y+\Delta y) - \psi(x_2,y) >
\varepsilon/2>0$ if $|\Delta y|< \delta$. At the same time
$(\psi(x_1,y+\Delta y) - \psi(x_2,y)) \Delta y \leq 0$, hence, for
a positive $0<\Delta y< \delta$ we have a contradiction.
Therefore, the function $\frac{\partial f(p,p^{*})}{\partial p}$
is a monotonic function of $\frac{p}{p^{*}}$, hence,
$f(p,p^{*})=p^{*}h(\frac{p}{p^{*}})+{\rm const}(p^*)$, where $h$
is a convex function of one variable.
\end{proof}

This lemma has important corollaries about many popular
divergences $H(P(t) \| P^*)$ which are not Lyapunov functions
of Markov chains. This means that there exist such
distributions $P_0$ and $P^*$ and a Markov chain with
equilibrium distribution $P^*$ that due to the Kolmogorov
equations $$\left. \frac{\D H(P(t) \| P^*)}{\D t}\right|_{t=0}>
0$$ if $P(0)=P_0$. This Markov process increases divergence
between the distributions $P,P^*$ (in a vicinity of $P_0$)
instead of making them closer. For example,

\begin{corollary}The following Bregman divergences \cite{Bregman1967} are not
universal Lyapunov functions for Markov chains:
\begin{itemize}
\item{Squared Euclidean distance $B(P \| P^*)=\sum_i
(p_i-p^*_i)^2$;}
\item{The Itakura--Saito divergence \cite{Itakura1968} $B(P \| P^*)=\sum_i
\left(\frac{p_i}{p_i^*}-\log \frac{p_i}{p_i^*}-1\right)$. $\ \ \
\square$}
\end{itemize}
\end{corollary}
These divergences violate the requirement: due to the Markov
process distributions always monotonically approach equilibrium.
(Nevertheless, among the Bregman divergences there exists a
universal Lyapunov function for Markov chains, the
Kulback--Leibler divergence.)

We place the proof of Theorem \ref{Theorem1NoMore} in Appendix.

\vspace{2mm}\noindent {\bf Remark}. If we relax the requirement of
smoothness and consider  in conditions of Theorem
\ref{Theorem1NoMore}  just continuous functions, then we have to
add to the answer the limit divergences, $$H_{\rm CR \ \infty}(P
\| P^*) =\max_i\left\{\frac{p_i}{p_i^*}\right\}-1 \ ;$$
 $$H_{\rm CR \ -\infty} (P \| P^*)= \max_i\left\{\frac{p_i^*}{p_i}\right\}-1 $$

\section{Markov Order \label{MarOd}}

\subsection{Entropy: a Function or an Order? \label{FunctionOrOrder}}

Theorem \ref{Theorem1NoMore} gives us all of the divergences for
which (i) the Markov chains monotonically approach their
equilibrium, (ii) the level sets are the same as for a separable
(sum over states) divergence and (iii) the level sets are the same
as for a divergence which is additive with respect to union of
independent subsystems.

We operate with the level sets and their orders, compare where the
divergence is larger (for monotonicity of the Markov chains
evolution), but the values of entropy are not important by
themselves. We are interested in the following order: $P$ precedes
$Q$ with respect to the divergence $H_{\ldots}(P \| P^*)$ if there
exists such a continuous curve $P(t)$ ($t\in [0,1]$) that
$P(0)=P$, $P(1)=Q$ and the function $H(t)=H_{\ldots}(P(t) \| P^*)$
monotonically decreases on the interval $t\in[0,1]$. This property
is invariant with respect to a monotonic (increasing)
transformation of the divergence. Such a transformation does not
change the conditional minimizers or maximizers of the divergence.

There exists one important property that is not invariant with
respect to monotonic transformations. The increasing function
$F(H)$ of a convex function $H(P)$ is not obligatorily a convex
function. Nevertheless, the sublevel sets given by inequalities
$H(P) \leq a$ coincide with the sublevel sets $F(H(P)) \leq F(a)$.
Hence, sublevel sets for $F(H(P))$ remain convex.

The Jensen inequality $$H(\theta P + (1-\theta)Q) \leq \theta H(P) + (1-\theta)H(Q)$$
($\theta \in [0,1]$) is not invariant with respect to monotonic transformations. Instead
of them, there appears the {\em max form analogue of the Jensen inequality}
(quasiconvexity \cite{Sion1958}):

\begin{equation}\label{maxJen}
H(\theta P + (1-\theta)Q) \leq \max\{H(P), H(Q)\} \ , \ \ \theta \in
[0,1]
\end{equation}
This inequality is invariant with respect to monotonically
increasing transformations and it is equivalent to convexity of
sublevel sets.

\begin{proposition}\label{sublevelconvex1}All sublevel sets of a function $H$ on a convex set $V$ are
convex if and only if for any two points $P,Q\in V$ and every
$\theta \in [0,1]$ the inequality (\ref{maxJen}) holds. \;\;\;
\;\;\; $\square$
\end{proposition}

It seems very natural to consider divergences as orders on
distribution spaces, and discuss only properties which are
invariant with respect to monotonic transformations. From this
point of view, the CR family appears absolutely naturally from the
additivity (ii) and the ``sum over states" (iii) axioms, as well
as the convex combination $\beta D_{\mathrm{KL}}(P\|P^*) +
(1-\beta)D_{\mathrm{KL}}(P^*\|P)$ ($\alpha \in [0,1]$), and in the
above property context there are no other smooth divergences.

\subsection{Description of Markov Order \label{order}}

The CR family and the convex combinations of Shannon's and Burg
relative entropies are distinguished families of divergences.
Apart from them there are many various ``divergences", and even
the Csisz\'ar--Morimoto functions (\ref{Morimoto}) do not include
all used possibilities. Of course, most users prefer to have an
unambiguous choice of entropy: it would be nice to have ``the best
entropy" for any class of problems. But from some point of view,
ambiguity of the entropy choice is unavoidable. In this section we
will explain why the choice of entropy is necessarily non unique
and demonstrate that for many MaxEnt problems the natural solution
is not a fixed distribution, but a well defined set of
distributions.

The most standard use of divergence in many application is as
follows:
\begin{enumerate}
\item{On a given space of states an ``equilibrium distribution"
$P^*$ is given. If we deal with the probability distribution in real
kinetic processes then it means that without any additional
restriction the current distribution will relax to $P^*$. In that
sense, $P^*$ is the most disordered distribution. On the other hand,
$P^*$ may be considered as the ``most disordered" distribution with
respect to some a priori information.}
 \item{We do not know the current distribution $P$, but
we do know some linear functionals, the moments $u(P)$.}
 \item{We do not want to introduce any subjective arbitrariness in the
estimation of $P$ and define it as the ``most disordered"
distribution for given value $u(P)=U$ and equilibrium $P^*$. That
is, we define $P$ as solution to the problem:
\begin{equation}\label{MaxEntCOnd}
H_{\ldots}(P \| P^*) \to \min \ \ {\rm subject \ to} \ \ u(P)=U
\end{equation}
Without the condition $u(P)=U$ the solution should be simply
$P^*$.}
\end{enumerate}

Now we have too many entropies and do not know what is the
optimal choice of $H_{\ldots}$ and what should be the optimal
estimate of $P$. In this case the proper question may be: {\it
which $P$ could not be such an optimal estimate}? We can answer
the exclusion question. Let for a given $P^0$ the condition
hold, $u(P^0)=U$. If there exists a Markov process with
equilibrium $P^*$ such that at point $P^0$ due to the
Kolmogorov equation (\ref{MAsterEq1}) $$\frac{\D P}{\D t} \neq
0 \ \ {\rm and} \ \ \frac{\D (u(P))}{\D t} = 0 \ $$ then $P^0$
cannot be the optimal estimate of the distribution $P$ under
condition $u(P)=U$.

The motivation of this approach is simple: any Markov process
with equilibrium $P^*$ increases disorder and  brings the
system ``nearer" to the equilibrium $P^*$. If at $P^0$ it is
possible to move along the condition plane towards the more
disordered distribution then $P^0$ cannot be considered as an
extremely disordered distribution on this plane. On the other
hand, we can consider $P^0$ as a possible extremely disordered
distribution on the condition plane, if for any Markov process
with equilibrium $P^*$ the solution of the Kolmogorov equation
(\ref{MAsterEq1}) $P(t)$ with initial condition $P(0)=P^0$ has
no points on the plane $u(P)=U$ for $t>0$.

Markov process here is considered as a ``randomization". Any set
$C$ of distributions can be divided in two parts: the
distributions which retain in $C$ after some non-trivial
randomization and the distributions which leave $C$ after any
non-trivial randomization. The last are the maximally random
elements of $C$: they cannot become more random and retain in $C$.
Conditional minimizers of relative entropies $H_h(P\|P^*)$ in $C$
are maximally random in that sense.

There are too many functions $H_h(P\|P^*)$ for effective
description of all their conditional minimizers. Nevertheless, we
can describe the maximally random distributions directly, by
analysis of Markov processes.

To analyze these properties more precisely, we need some formal
definitions.

\begin{definition}  (Markov preorder). If for distributions $P^0$
and $P^1$ there exists such a Markov process with equilibrium
$P^*$ that for the solution of the Kolmogorov equation with
$P(0)=P^0$ we have $P(1)=P^1$ then we say that $P^0$ and $P^1$
are connected by the Markov preorder with equilibrium $P^*$ and
use notation $P^0 \succ^0_{P^*} P^1$.
\end{definition}

\begin{definition} Markov order is the closed transitive closure of the Markov
preorder. For the Markov order with equilibrium $P^*$ we use
notation $P^0 \succ_{P^*} P^1$.
\end{definition}

For a given $P^*=(p^*_i)$ and a distribution $P=(p_i)$ the set of
all vectors $v$ with coordinates $$v_i= \sum_{j, \, j\neq i}
q_{ij}p^*_j\left(\frac{p_j}{p_j^*}-\frac{p_i}{p_i^*}\right)$$
where $p_i^*$ and $q_{ij} \geq 0$ are connected by identity
(\ref{MasterEquilibrium}) is a closed convex cone. This is a cone
of all possible time derivatives of the probability distribution
at point $P$ for Markov processes with equilibrium $P^*=(p^*_i)$.
For this cone, we use notation $\mathbf{Q}_{(P,P^*)}$

\begin{definition}For each distribution $P$ and a $n$-dimensional vector
$\Delta$ we say that $\Delta <_{(P,P^*)}0$ if $\Delta \in
\mathbf{Q}_{(P,P^*)}$. This is the {\it local Markov order}.
\end{definition}

\begin{proposition}\label{propercone2}$\mathbf{Q}_{(P,P^*)}$ is a proper cone, \emph{i.e.}, it does not
include any straight line.
\end{proposition}

\begin{proof} To prove this proposition its is sufficient to analyze the
formula for entropy production (for example, in form
(\ref{ENtropyProd})) and mention that for strictly convex $h$ (for
example, for traditional $x\ln x$ or $(x-1)^2/2$) $\D H_h / \D t
=0$ if and only if $\D P / \D t =0$. If the cone
$\mathbf{Q}_{(P,P^*)}$ includes both vectors $x$ and $-x$ ($x\neq
0$ it means that there exist Markov chains with equilibrium $P^*$
and with opposite time derivatives at point $P$. Due to the
positivity of entropy production (\ref{ENtropyProd}) this is
impossible.
\end{proof}

The connection between the local Markov order and the Markov order
gives the following proposition, which immediately follows from
definitions.

\begin{proposition}\label{smoothcurve3}$P^0 \succ_{P^*} P^1$ if and only if there exists such
a continuous almost everywhere differentiable curve $P(t)$ in the
simplex of probability distribution that $P(0)=P^0$, $P(1)=P^1$
and for all $t\in [0,1]$, where $P(t)$ is differentiable,
\begin{equation}\label{MarkovInclusion}
\frac{\D P(t)}{\D t}  \in \mathbf{Q}_{(P(t),P^*)} \; \; \;
\;\;\;\square
\end{equation}
\end{proposition}

For our purposes, the following estimate of the Markov order
through the local Markov order \linebreak is important.

\begin{proposition}\label{Dominantlocal4}If $P^0 \succ_{P^*} P^1$ then $P^0
>_{(P^0,P^*)} P^1$, \emph{i.e.}, $P^1-P^0 \in \mathbf{Q}_{(P,P^*)}$.
\end{proposition}

This proposition follows from the characterization of the local
order and detailed description of the cone
$\mathbf{Q}_{(P(t),P^*)}$ (Theorem \ref{Theorem2detbalsuff}
below).

Let us recall that a convex pointed  cone is a convex envelope of its extreme rays. A ray
with directing vector $x$ is a set of points $\lambda x$ ($\lambda \geq 0$). We say that
$l$ is an extreme ray of $\mathbf{Q}$ if for any $u \in l$ and any $x,y \in \mathbf{Q}$,
whenever $u = (x + y)/2$, we must have $x,y\in l$. To characterize the extreme rays of
the cones of the local Markov order $\mathbf{Q}_{(P,P^*)}$ we need a graph representation
of the Markov chains. We use the notation $A_i$ for states (vertices), and designate
transition from state $A_i$ to state $A_j$ by an arrow (edge) $A_i \to A_j$. This
transition has its transition intensity $q_{ji}$ (the coefficient in the Kolmogorov
equation (\ref{MAsterEq0})).

\begin{lemma}Any extreme ray of the cone $\mathbf{Q}_{(P,P^*)}$
corresponds to a Markov process which transition graph is a simple
cycle $$A_{i_1} \to A_{i_2} \to \ldots A_{i_k} \to A_{i_1}$$ where
$k \leq n$, all the indices $i_1,\ldots i_k$ are different, and
transition intensities for a directing vector of such an extreme ray
$q_{i_{j+1} \ i_{j} }$ may be selected as $1/p_{i_j}^*$:
\begin{equation}\label{CycleRates}
q_{i_{j+1} \ i_{j} }=\frac{1}{p_{i_j}^*}
\end{equation}
(here we use the standard convention that for a cycle $q_{i_{k+1}
\ i_{k} }=q_{i_{1} \ i_{k} }$).
\end{lemma}

\begin{proof} First of all, let us mention that if for
three vectors $x,y,u \in \mathbf{Q}_{(P,P^*)}$ we have $u = (x +
y)/2$ then the set of transitions with non-zero intensities for
corresponding Markov processes for $x$ and $y$ are included in this
set for $u$ (because negative intensities are impossible). Secondly,
just by calculation of the free variables in the equations
(\ref{MasterEquilibrium}) (with additional condition) we find that
the the amount of non-zero intensities for a transition scheme which
represents an extreme ray should be equal to the amount of states
included in the transition scheme. Finally, there is only one scheme
with $k$ vertices, $k$ edges and a positive equilibrium, a simple
oriented cycle.\end{proof}

\begin{theorem}\label{Theorem2detbalsuff}Any extreme ray of the cone $\mathbf{Q}_{(P,P^*)}$
corresponds to a Markov process whose transition graph is a simple
cycle of the length 2: $A_i \rightleftarrows A_j$. A transition
intensities $q_{ij}, \ q_{ji}$ for a directing vector of such an
extreme ray may be selected as
\begin{equation}\label{2Rates}
q_{ij} = \frac{1}{p^*_j} \ , \ \ q_{ji} = \frac{1}{p^*_ i}
\end{equation}
\end{theorem}

\begin{proof} Due to Lemma 2, it is sufficient to prove
that for any distribution $P$ the right hand side of the
Kolmogorov equation (\ref{MAsterEq1}) for a simple cycle with
transition intensities (\ref{CycleRates}) is a conic
combination (the combination with non-negative real
coefficients) of the right hand sides of this equation for
simple cycles of the length 2 at the same point $P$. Let us
prove this by induction. For the cycle length 2 it is trivially
true. Let this hold for the cycle lengths $2,\ldots n-1$. For a
cycle of length $n$, $A_{i_1} \to A_{i_2} \to \ldots A_{i_k}
\to A_{i_1}$, with transition intensities given by
(\ref{CycleRates}) the right hand side of the Kolmogorov
equation is the vector $v$ with coordinates
$$v_{i_j}=\frac{p_{i_{j-1}}}{p^*_{i_{j-1}}}
-\frac{p_{i_{j}}}{p^*_{i_{j}}}$$ (under the standard convention
regarding cyclic order). Other coordinates of $v$ are zeros. Let
us find the minimal value of ${p_{i_{j}}}/{p^*_{i_{j}}}$ and
rearrange the indices by a cyclic permutation to put this minimum
in the first place:
$$\min_j\left\{\frac{p_{i_{j}}}{p^*_{i_{j}}}\right\}=
\frac{p_{i_{1}}}{p^*_{i_{1}}} $$
 The vector $v$ is a sum of two vectors: a directing vector for
 the cycle $A_{i_2} \to \ldots A_{i_k} \to A_{i_2}$ of the length
 $n-1$ with transition intensities given by formula
 (\ref{CycleRates}) (under the standard convention
about the cyclic order for this cycle) and a vector
$$\frac{\frac{p_{i_n}}{p_{i_n}^*}-\frac{p_{i_1}}{p_{i_1}^*}}
{\frac{p_{i_2}}{p_{i_2}^*}-\frac{p_{i_1}}{p_{i_1}^*}} v^2$$ where
$v^2$ is the directing vector for a cycle of length 2,
$A_{i_1}\rightleftarrows A_{i_2}$ which can have only two \linebreak non-zero
coordinates: $$v^2_{i_1}=\frac{p_{i_{2}}}{p^*_{i_{2}}}
-\frac{p_{i_{1}}}{p^*_{i_{1}}}=-v^2_{i_2}$$ The coefficient in
front of $v^2$ is positive because ${p_{i_{1}}}/{p^*_{i_{1}}}$ is
the minimal value of ${p_{i_{j}}}{p^*_{i_{j}}}$. A case when
${p_{i_{1}}}/{p^*_{i_{1}}} ={p_{i_{2}}}/{p^*_{i_{2}}}$ does not
need special attention because it is equivalent to the shorter
cycle $A_{i_1} \to A_{i_3} \to \ldots A_{i_k} \to A_{i_1}$
($A_{i_2}$ could be omitted). A conic combination of conic
combinations is a conic combination again.\end{proof}

It is quite surprising that the local Markov order and, hence, the
Markov order also are generated by the reversible Markov chains
which satisfy the detailed balance principle. We did not include any
reversibility assumptions, and studied the general Markov chains.
Nevertheless, for the study of orders, the system of cycles of
length 2 all of which have the same equilibrium is sufficient.

\subsection{Combinatorics of Local Markov Order}

Let us describe the local Markov order in more detail. First of all,
we represent kinetics of the reversible Markov chains. For each pair
$A_i, A_j$ ($i\neq j$) we select an arbitrary order in the pair and
write the correspondent cycle of the length 2 in the form $A_i
\leftrightarrows A_j$. For this cycle we introduce the directing
vector $\gamma^{ij}$ with coordinates
\begin{equation}\label{gamma}
\gamma^{ij}_k=-\delta_{ik}+\delta_{jk}
\end{equation}
 where $\delta_{ik}$ is the
Kronecker delta. This vector has the $i$th coordinate $-1$, the
$j$th coordinate $1$ and other coordinates are zero. Vectors
$\gamma^{ij}$ are parallel to the edges of the standard simplex in
$R^n$. They are antisymmetric in their indexes:
$\gamma^{ij}=-\gamma^{ji}$.

We can rewrite the Kolmogorov equation in the form
\begin{equation}\label{MasterRevers}
\frac{\D P}{\D t}=\sum_{{\rm pairs}\ ij} \gamma^{ij} w_{ji}
\end{equation}
where $i\neq j$, each pair is included in the sum only once (in
the preselected order of $i,j$) and
 $$w_{ji}=r_{ji}\left(\frac{p_i}{p_i^*}-\frac{p_j}{p_j^*}\right)$$
The coefficient $r_{ji} \geq 0$ satisfies the detailed balance
principle: $$r_{ji}=q_{ji}p_i^*=q_{ij}p_j^*=r_{ij}$$

We use the three-value sign function:
\begin{equation}
{\rm sign} x=\left\{ \begin{array}{ll}
  -1, \ &{\rm if} \ x<0 ; \\
  0, \ &{\rm if} \ x=0 ; \\
  1,  \ &{\rm if} \ x>0
 \end{array} \right.
\end{equation}
With this function we can rewrite Equation (\ref{MasterRevers})
again as follows:
\begin{equation}\label{MasterReversSign}
\frac{\D P}{\D t}=\sum_{{\rm pairs}\ ij, \ r_{ji}\neq 0 }r_{ji}
\gamma^{ij}{\rm sign}
\left(\frac{p_i}{p_i^*}-\frac{p_j}{p_j^*}\right)
\left|\frac{p_i}{p_i^*}-\frac{p_j}{p_j^*}\right|
\end{equation}
The non-zero coefficients $r_{ji}$ may be arbitrary positive
numbers. Therefore, using Theorem \ref{Theorem2detbalsuff}, we
immediately find that the cone of the local Markov order at point
$P$ is
\begin{equation}\label{LocalOrderCone}
\mathbf{Q}_{(P,P^*)}={\rm cone}\left\{\gamma^{ij} {\rm sign}
\left.\left(\frac{p_i}{p_i^*}-\frac{p_j}{p_j^*}\right) \  \right| \
r_{ji}>0  \right\}
\end{equation}
where cone$\{\}$ stands for the conic hull.

The number ${\rm sign}
\left(\frac{p_i}{p_i^*}-\frac{p_j}{p_j^*}\right)$ is 1, when
$\frac{p_i}{p_i^*}>\frac{p_j}{p_j^*}$, $-1$, when
$\frac{p_i}{p_i^*}<\frac{p_j}{p_j^*}$ and 0, when
$\frac{p_i}{p_i^*}=\frac{p_j}{p_j^*}$. For a given $P^*$, the
standard simplex of distributions $P$ is divided by planes
$\frac{p_i}{p_i^*}=\frac{p_j}{p_j^*}$ into convex polyhedra where
functions ${\rm sign}
\left(\frac{p_i}{p_i^*}-\frac{p_j}{p_j^*}\right)$ are constant. In
these polyhedra the cone of the local Markov order
(\ref{LocalOrderCone}) $\mathbf{Q}_{(P,P^*)}$ is also constant. Let
us call these polyhedra {\it compartments}.

\begin{figure}[t]
\caption{\leftskip=1cm \rightskip=1cm \label{3stateCones}
Compartments $\mathcal{C}_{\sigma}$, corresponding cones
$\mathbf{Q}_{\sigma}$ (the angles) and all tableaus $\sigma$ for
the Markov chain with three states (the choice of equilibrium
($p_i^*=1/3$), does not affect combinatorics and topology of
tableaus, compartments and cones).}}
\centering{\includegraphics[width=130mm]{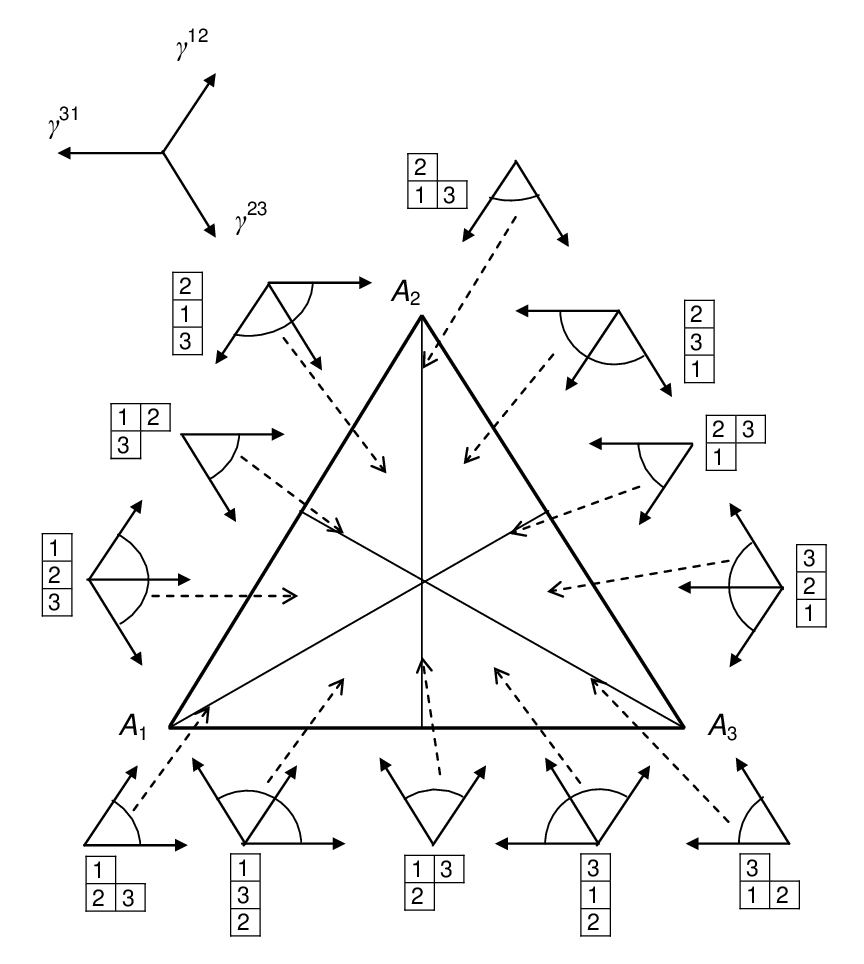}
\end{figure}

In Figure~\ref{3stateCones} we represent compartments and cones of
the local Markov order for the Markov chains with three states,
$A_{1,2,3}$. The reversible Markov chain consists of three
reversible transitions $A_1 \leftrightarrows A_2\leftrightarrows
A_3\leftrightarrows A_1$ with corresponding directing vectors
$\gamma^{12} = (-1, 1, 0)^\top$; $\gamma^{23} = (0,-1, 1)^\top$;
$\gamma^{31} = (1, 0,-1)^\top$. The topology of the partitioning of
the standard simplex into compartments and the possible values of
the cone $\mathbf{Q}_{(P,P^*)}$ do not depend on the position of the
equilibrium distribution $P^*$.

Let us describe all possible compartments and the correspondent
local Markov order cones. For every natural number $k \leq n-1$
the $k$-dimensional compartments are numerated by surjective
functions $\sigma :\{1,2,\ldots ,n\} \to \{1,2,\ldots ,k+1\}$.
Such a function defines the partial ordering of quantities
$\frac{p_j}{p_j^*}$ inside the compartment:
\begin{equation}
\frac{p_i}{p_i^*} > \frac{p_j}{p_j^*} \; \; {\rm if} \;\; \sigma(i)
< \sigma (j); \;\;\; \frac{p_i}{p_i^*} = \frac{p_j}{p_j^*} \; \;
{\rm if} \;\; \sigma(i) = \sigma (j)
\end{equation}

Let us use for the correspondent compartment notation
$\mathcal{C}_{\sigma}$ and for the Local Markov order cone
$Q_{\sigma}$. Let $k_i$ be a number of elements in preimage of $i$
($i=1,\ldots , k$): $k_i=|\{j\ | \ \sigma (j)=i \}|$. It is
convenient to represent surjection $\sigma$ as a tableau with $k$
rows and $k_i$ cells in the $i$th row filled by numbers from
$\{1,2,\ldots ,n\}$. First of all, let us draw diagram, that is a
finite collection of cells arranged in left-justified rows. The
$i$th row has $k_i$ cells. A tableau is obtained by filling cells
with numbers $\{1,2,\ldots ,n\}$. Preimages of $i$ are located in
the $i$th row. The entries in each row are increasing. (This is
convenient to avoid ambiguity of the representation of the
surjection $\sigma$ by the diagram.) Let us use for tableaus the
same notation as for the corresponding surjections.

Let a tableau $A$ have $k$ rows. We say that a tableau $B$ follows
$A$ (and use notation $A \to B$) if $B$ has $k-1$ rows and $B$ can
be produced from $A$ by joining of two neighboring rows in $A$ (with
ordering the numbers in the joined row). For the transitive closure
of the relation $\to$ we use notation $\Rrightarrow$.

\begin{proposition}\label{relabound5}$r \partial Q_{\sigma}=
\bigcup_{\sigma \Rrightarrow  \varsigma} Q_{\varsigma} \; \;
\;\;\; \; \square$
\end{proposition}
Here $r \partial U$ stands for the ``relative boundary" of a set $U$
in the minimal linear manifold which includes $U$.

The following Proposition characterizes the local order cone through
the surjection $\sigma$. It is sufficient to use in definition of
$Q_{\sigma}$ (\ref{LocalOrderCone}) vectors $\gamma^{ij}$
(\ref{gamma}) with $i$ and $j$ from the neighbor rows of the diagram
(see Figure~\ref{3stateCones}).

\begin{proposition}\label{extremerays6} For a given surjection $\sigma$ compartment
$\mathcal{C}_{\sigma}$ and cone $Q_{\sigma}$ have the following
description:
\begin{equation}\label{CompartDescr}
\mathcal{C}_{\sigma}=\left\{ P \ | \
\frac{p_i}{p^*_i}=\frac{p_j}{p^*_j} \;\; {\rm for} \;\;
\sigma(i)=\sigma(j) \;\;{\rm and} \;\;
\frac{p_i}{p^*_i}>\frac{p_j}{p^*_j}  \;\; {\rm for} \;\;
\sigma(j)=\sigma(i)+1 \right\}
\end{equation}
\begin{equation}\label{ConeDescr}
Q_{\sigma}={\rm cone}\{\gamma^{ij} \ | \ \sigma(j)=\sigma(i)+1 \}
\;\;\;\;\; \; \square
\end{equation}
\end{proposition}

Compartment $\mathcal{C}_{\sigma}$ is defined by equalities
$\frac{p_i}{p^*_i}=\frac{p_j}{p^*_j}$ where $i,j$ belong to one row
of the tableau $\sigma$ and inequalities
$\frac{p_i}{p^*_i}>\frac{p_j}{p^*_j}$ where $j$ is situated in a row
one step down from $i$ in the tableau ($\sigma(j)=\sigma(i)+1$).
Cone $Q_{\sigma}$ is a conic hull of $\sum_{i=1}^{k-1} k_i k_{i+1}$
vectors $\gamma^{ij}$. For these vectors, $j$ is situated in a row
one step down from $i$ in the tableau.  Extreme rays of $Q_{\sigma}$
are products of the positive real half-line on vectors $\gamma^{ij}$
(\ref{ConeDescr}).

Each compartment has the {\em lateral faces} and the {\em base}. We
call the face a lateral face, if its closure includes the
equilibrium $P^*$. The base of the compartment belongs to a border
of the standard simplex of probability distributions.

To enumerate all the lateral faces of a $k$-dimensional compartment
$\mathcal{C}_{\sigma}$ of codimension $s$ (in
$\mathcal{C}_{\sigma}$) we have to take all subsets with $s$
elements in $\{1,2,\ldots ,k\}$. For any such a subset $J$ the
correspondent $k-s$-dimensional lateral face is given by additional
equalities $\frac{p_i}{p_i^*} = \frac{p_j}{p_j^*}$ for
$\sigma(j)=\sigma(i)+1$, $i\in J$.

\begin{proposition}\label{ComartFaces}All $k-s$-dimensional lateral faces of a
$k$-dimensional compartment $\mathcal{C}_{\sigma}$ are in bijective
correspondence with the $s$-element subsets $J\subset \{1,2,\ldots
,k\}$. For each $J$ the correspondent lateral face is given in
$\mathcal{C}_{\sigma}$ by equations
\begin{equation}
\frac{p_i}{p_i^*} = \frac{p_j}{p_j^*} \; \;{\rm for \, all}
\;\;i\in J \;\; {\rm and} \;\;\sigma(j)=\sigma(i)+1 \;\; \;\;\; \;
\square
\end{equation}
\end{proposition}

The 1-dimensional lateral faces (extreme rays) of compartment
$\mathcal{C}_{\sigma}$ are given by selection of one number from
$\{1,2,\ldots ,k\}$ (this number is the complement of $J$). For this
number $r$, the correspondent 1-dimensional face is a set
parameterized by a positive number $a\in ]1,a_r]$,
$a_r=1/\sum_{\sigma(i)\leq r} p^*_i$:
\begin{equation}\label{compartmentRays}
\begin{split}
\frac{p_i}{p_i^*} =a, \; \; {\rm for} \;\;
 \sigma(i) \leq r \, ; \;\;
 \frac{p_i}{p_i^*}=b, \; \; {\rm for} \;\;
 \sigma(i) > r\, ; \\
a>1>b\geq 0, \; a\sum_{i, \, \sigma(i)\leq r} p^*_i+b\sum_{i, \,
\sigma(i)> r}p^*_i=1
\end{split}
\end{equation}

The compartment $\mathcal{C}_{\sigma}$ is the interior of the
$k$-dimensional simplex with vertices $P^*$ and $v_r$ ($r=1,2,\ldots
k$). The vertex $v_r$ is the intersection of the correspondent
extreme ray (\ref{compartmentRays}) with the border of the standard
simplex of probability distributions: $P=v_r$ if
\begin{equation}\label{verticesC}
p_i=p_i^*a_r, \;\; {\rm for}\;\; \sigma(i)\leq r ; \;\;
p_i=0\;\;{\rm for}\;\; \sigma(i)> r
\end{equation}

The base of the compartment $\mathcal{C}_{\sigma}$ is a
$k-1$-dimensional simplex with vertices $v_r$ ($r=1,2,\ldots k$).

It is necessary to stress that we use the reversible Markov chains
for construction of the general Markov order due to Theorem
\ref{Theorem2detbalsuff}.

\section{The ``Most Random" and Conditionally Extreme Distributions}

\subsection{Conditionally Extreme Distributions in Markov Order
\label{OrderMaxEnt}}

The Markov order can be used to reduce the uncertainty in the
standard settings. Let the plane $L$ of the known values of some
moments be given: $u^i(P)=U_i$ on $L$. Assume also that the
``maximally disordered" distribution (equilibrium) $P^*$ is known
and we assume that the probability distribution is $P^*$ if there
is no restrictions. Then, the standard way to evaluate $P$ for
given moment conditions $u^i(P)=U_i$ is known: just to minimize
$H_{\ldots}(P \| P^*)$ under these conditions. For the Markov
order we also can define the {\it conditionally extreme points} on
$L$.

\begin{definition}Let $L$ be an affine subspace of $\mathbf{R}^n$, $\Sigma _n$
be a standard simplex in $\mathbf{R}^n$. A probability
distribution $P\in L \cap \Sigma _n$ is a conditionally extreme
point of the Markov order on $L$ if $$(P+\mathbf{Q}_{(P,P^*)})
\cap L = \{P\}$$
\end{definition}

It is useful to compare this definition to the condition of the
extremum of a differentiable function $H$ on $L$: ${\rm grad} H
\bot L$.

\begin{figure}[t]
 \caption{\leftskip=1cm \rightskip=1cm \label{CondExtr1coinside} If the moments
are just some of $p_i$ then all points of conditionally minimal
divergence are the same for all the main divergences and coincide
with the unique conditionally extreme point of the Markov order
(example for the Markov chain with three states, symmetric
equilibrium ($p_i^*=1/3$)) and the moment plane $p_2=$const.}}
\centering{\includegraphics[width=110mm]{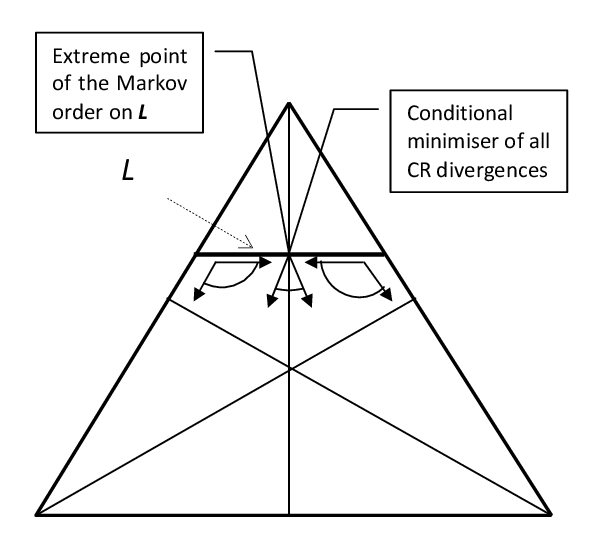}
\end{figure}

First of all, it is obvious that in the case when all the moments
$u^i(P)$ are just some of the values $p_i$, then there exists only
one extreme point of the Markov order on $L$, and this point is, at
the same time, the conditional minimum on $L$ of all
Csisz\'ar--Morimoto functions $H_h(P)$ (\ref{Morimoto}) (see, for
example, Figure~\ref{CondExtr1coinside}). This situation is
unstable, and for a small perturbation of $L$ the set of extreme
points of the Markov order on $L$ includes the intersection of $L$
with one of compartments (Figure~\ref{CondExtrdDifferent}a). For the
Markov chains with three states, each point of this intersection is
a conditional minimizer of one of the CR divergences (see
Fig.~\ref{CondExtrdDifferent}a). Such a situation persists for all
$L$ in general positions (Figure~\ref{CondExtrdDifferent}b). The
extreme points of the family $\beta D_{\mathrm{KL}}(P\|P^*) +
(1-\beta)D_{\mathrm{KL}}(P^*\|P)$ form an interval which is strictly
inside the interval of the extreme points of the Markov order on
$L$. For higher dimensions of $L \cap \Sigma _n$  the Markov order
on $L$ also includes the intersection of $L$ with some compartments,
however the conditional minimizers of the CR divergences form a
curve there, and extreme points of the family $\beta
D_{\mathrm{KL}}(P\|P^*) + (1-\beta)D_{\mathrm{KL}}(P^*\|P)$ on $L$
form another curve. These two curves intersect at two points
($\lambda=0,-1$), which correspond to the BGS and Burg relative
entropies.

\begin{figure}[h!]
 \caption{\leftskip=1cm \rightskip=1cm \label{CondExtrdDifferent}The set of
conditionally extreme points of the Markov order on the moment
plane in two general positions. For the main divergences the
points of conditionally minimal divergence are distributed in this
set. For several of the most important divergences these
minimizers are pointed out. In this simple example each extreme
point of the Markov order is at the same time a minimizer of one
of the $H_{\rm CR \ \lambda}$ ($\lambda \in ]-\infty, +\infty[$)
(examples for the Markov chain with three states, symmetric
equilibrium ($p_i^*=1/3$)).}}
\centering{a)\includegraphics[width=110mm]{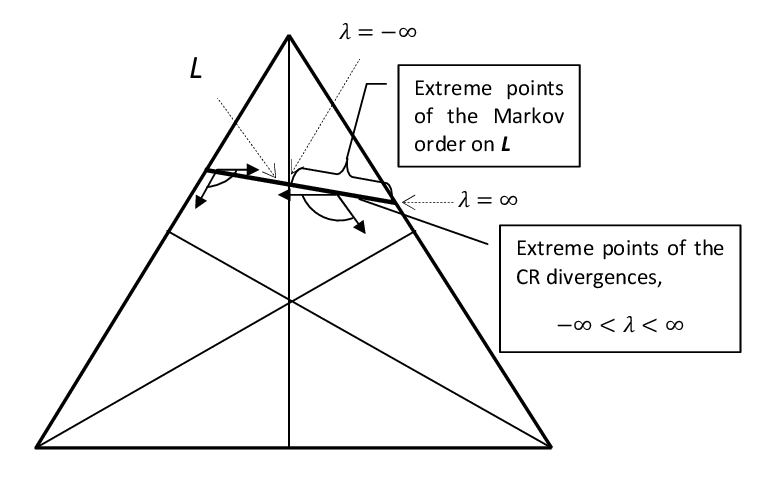}
b)\includegraphics[width=110mm]{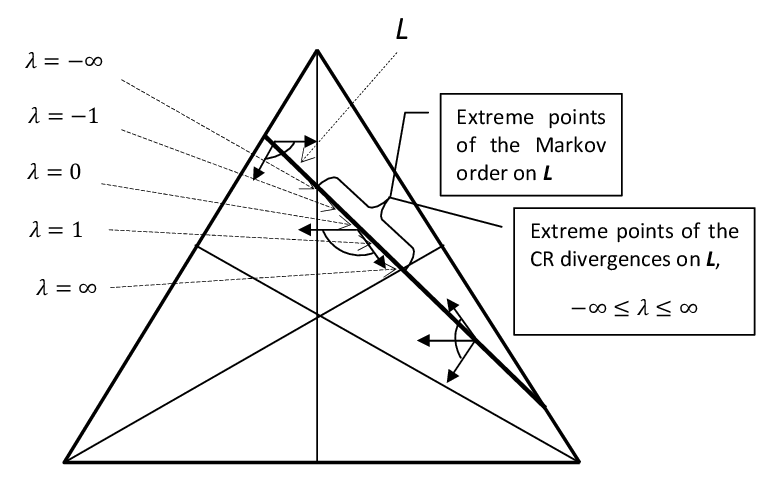}
\end{figure}

\newpage
\subsection{How to Find the Most Random Distributions? \label{CondExtrOrder}}

Let the plane $L$ of the known values of some moments be given:
$u^i(P)=\sum_j u^i_j p_j=U_i$ ($i=1,\ldots m$) on $L$. For a given
divergence $H(P \| P^*)$ we are looking for a conditional
minimizer $P$:
\begin{equation}\label{MaxEntProblem} H(P \| P^*) \to \min \ \
{\rm subject \ to} \ \ u^i(P)=U_i \\ (i=1,\ldots m)
\end{equation}
We can assume that $H(P \| P^*)$ is convex. Moreover, usually it
is one of the Csisz\'ar--Morimoto \newline functions (\ref{Morimoto}). This
is very convenient for numerical minimization because the matrix
of second derivatives is diagonal. Let us introduce the Lagrange
multipliers $\mu_i$ ($i=1,\ldots m$) and write the system of
equations ($\mu_0$ is the Lagrange multiplier for the total
probability identity $\sum_j p_j=1$ :
\begin{equation}\label{LagrMultEq}
\begin{split}
&\frac{\partial H}{\partial p_j}=\mu_0+ \sum_{i=1}^m \mu_i u^i_j \
; \\
 &\sum_{j=1}^n u^i_j p_j=U_i \ ; \\
 &\sum_{j=1}^n p_j = 1
\end{split}
\end{equation}
Here we have $n+m+1$ equations for $n+m+1$ unknown variables
($p_j$, $\mu_i$, $\mu_0$).

Usually $H$ is a convex function with a diagonal matrix of second
variables and the method of choice for solution of this equation
(\ref{LagrMultEq}) is the Newton method. On the $l+1$st iteration to
find $P^{l+1}=P^l+\Delta P$ we have to solve the following system of
linear equations
\begin{equation}\label{NewtonStep}
\begin{split}
& \sum_{s=1}^n \left. \frac{\partial^2 H}{\partial p_j \partial
p_s }\right|_{P=P^l} \Delta p_s =\mu_0+ \sum_{i=1}^m \mu_i u^i_j -
\left. \frac{\partial H}{\partial p_j}\right|_{P=P^l} \ ;
\\
 &\sum_{j=1}^n u^i_j \Delta p_j=0 \ ; \\
 &\sum_{j=1}^n \Delta p_j = 0
\end{split}
\end{equation}
For a diagonal matrix of the second derivatives the first $n$
equations can be explicitly resolved. If for the solution of this
system (\ref{NewtonStep}) the positivity condition $p_j^l + \Delta
p_j > 0$ does not hold (for some of $j$) then we should decrease
the step, for example by multiplication $\Delta P := \theta \Delta
P$, where $$0< \theta <  \min_{p_i^l+\Delta p_i
<0}\frac{p_i^l}{|\Delta p_i |}$$

For initial approximation we can take any positive normalized
distribution which satisfies the conditions $u^i(P)=U_i$
($i=1,\ldots m$).

For the Markov orders the set of conditionally extreme
distributions consists of intersections of $L$ with compartments.

Here we find this set for one moment condition of the form
$u(P)=\sum_j u_j p_j=U$. First of all, assume that $U \neq U^*$,
where $U^*=u(P^*)=\sum_j u_j p^*_j$ (if $U = U^*$ then equilibrium
is the single conditionally extreme distribution). In this case,
the set of conditionally extreme distributions is the intersection
of the condition hyperplane with the closure of one compartment
and can be described by the following system of equations and
inequalities (under standard requirements $p_i \geq0$, $\sum_i p_i
=1$ ):
\begin{equation}\label{1condExtrPoints}
\begin{split}
&\sum_j u_j p_j=U; \\
  &\frac{p_i}{p_i^*}\geq \frac{p_j}{p_j^*}\;\; {\rm if} \;\; u_i(U-U^*)\geq u_j(U-U^*)
\end{split}
\end{equation}
(hence, $ \frac{p_i}{p_i^*}= \frac{p_j}{p_j^*}$  if $u_i=u_j$).

To find this solution it is sufficient to study dynamics of $u(P)$
due to equations (\ref{MasterRevers}) and to compare it with
dynamics of $u(P)$ due to a model system $\dot{P}=P^*-P$. This
model system is also a Markov chain and, therefore, $P^*-P\in
\mathbf{Q}_{(P,P^*)}$. Equations and inequalities
(\ref{1condExtrPoints}) mean that the set of conditionally extreme
distributions is the intersection of the condition hyperplane with
the closure of compartment $\mathcal{C}$. In $\mathcal{C}$,
numbers $\frac{p_i}{p_i^*}$ have the same order on the real line
as numbers $u_i(U-U^*)$ have, these two tuples of numbers
correspond to the same tableau $\sigma$ and $\mathcal{C}=
\mathcal{C}_{\sigma}$.

For several linearly independent conditions there exists a
condition plane $L$:
\begin{equation}\label{CondSpace}
u^i(P)=\sum_j u^i_j p_j=U_i \;\;(i=1,\ldots m)
\end{equation}
Let us introduce the $m$-dimensional space $T$ with coordinates
$u^i$. Operator $u(P)=(u^i(P))$ maps the distribution space into
$T$ and the affine manifold $L$ (\ref{CondSpace}) maps into a
point with coordinates $u^i=U_i$.

If $P^*\in L$ then the problem is trivial and the only extreme
distribution of the Markov order on  $L$ is $P^*$. Let us assume
that $P^*\notin L$.

For each distribution $P\in L$ we can study the possible direction
of motions of projection distributions onto $T$ due to the Markov
processes.

First of all, let us mention that if $u(\gamma^{ij})=0$ then the
transitions $A_i \leftrightarrows A_j$ move the distribution along
$L$. Hence, for any conditionally extreme distribution $P\in L$
this transition $A_i \leftrightarrows A_j$ should be in
equilibrium and the partial equilibrium condition holds:
$\frac{p_i}{p_i^*} = \frac{p_j}{p_j^*}$.

Let us consider processes with $u(\gamma^{ij}) \neq 0$. If there
exists a convex combination (\ref{MasterReversSign}) of vectors
$u(\gamma^{ij}){\rm sign}
\left(\frac{p_i}{p_i^*}-\frac{p_j}{p_j^*}\right)$ ($u(\gamma^{ij})
\neq 0$) that is equal to zero then $P$ cannot be an extreme
distribution of the Markov order on $L$.

These two conditions for vectors $\gamma^{ij}$ with
$u(\gamma^{ij})=0$ and for the set of vectors with non-zero
projection on the condition space define the extreme distributions
of the Markov order on the condition plane $L$ for several
conditions.

\section{Generalized Canonical Distribution \label{GCD}}

\subsection{Reference Distributions for Main Divergences}

A system with equilibrium $P^*$ is given and expected values of
some variables $u_j(P)=U_j$ are known. We need to find a
distribution $P$ with these values $u_j(P)=U_j$ and is ``the
closest" to the equilibrium distribution under this condition.

This distribution parameterized through expectation values is
often called the {\em reference distribution} or {\em generalized
canonical distribution}. After Gibbs and Jaynes, the standard
statement of this problem is an optimization problem: $$H(P\|P^*)
\to \min, \;\; u_j(P)=U_j$$ for appropriate divergence
$H(P\|P^*)$. If the number of conditions is $m$ then this
optimization problem can be often transformed into $m+1$ equations
with $m+1$ unknown Lagrange multipliers.

In this section, we study the problem of the generalized canonical
distributions for single condition $u(P)=\sum_{i=1}^n u_i p_i=U$,
$U\neq U^*$.

For the Csisz\'ar--Morimoto functions $H_h(P\|P^*)$
\begin{equation}
\frac{\partial H_h}{\partial p_i}=h'\left(\frac{p_i}{p_i^*}\right)
\end{equation}
We assume that the function $h'(x)$ has an inverse function $g$:
$g(h'(x))=x$ for any $x\in ]0,\infty[$. The method of Lagrange
multipliers gives for the generalized canonical distribution:

\begin{equation}
\frac{\partial H_h}{\partial p_i}=\mu_0\frac{\partial
(\sum_{j=1}^n p_j )}{\partial p_i}+\mu \frac{\partial U}{\partial
p_i}\, , \, h'\left(\frac{p_i}{p_i^*}\right)=\mu_0+\mu u_i, \,
\sum_{i=1}^n p_i=1, \, \sum_{i=1}^n p_i u_i =U
\end{equation}
As a result, we get the final expression for the distribution
$$p_i=p^*_i g(\mu_0+u_i \mu)$$ and equations for Lagrange
multipliers $\mu_0$ and $\mu$:
\begin{equation}\label{LagrangeEq}
\sum_{i=1}^n p^*_i g(\mu_0+u_i \mu)=1, \, \sum_{i=1}^n p^*_i
g(\mu_0+u_i \mu) u_i =U
\end{equation}
If the image of $h'(x)$ is the whole real line
($h'(]0,\infty[)=R$) then for any real number $y$ the value
$g(y)\geq 0$ is defined and there exist no problems about
positivity of $p_i$ due to (\ref{LagrangeEq}).

For the BGS relative entropy $h'(x)=\ln x$ (we use the normalized
$h(x)=x \ln x - (x-1)$ (\ref{normalization})). Therefore,
$g(x)=\exp x$ and for the generalized canonical distribution we
get
\begin{equation}\label{BGScanonical}
p_i=p^*_i {\rm e}^{\mu_0}{\rm e}^{u_i \mu}, \, {\rm
e}^{-\mu_0}=\sum_{i=1}^n p^*_i {\rm e}^{u_i \mu}, \,`\sum_{i=1}^n
p^*_i u_i  {\rm e}^{u_i \mu}=U \sum_{i=1}^n p^*_i {\rm e}^{u_i
\mu}
\end{equation}
As a result, we get one equation for $\mu$ and an explicit
expression for $\mu_0$ through $\mu$.

These $\mu_0$ and $\mu$ have the opposite sign comparing to
(\ref{boltzmannDistribution}) just because the formal difference
between the entropy maximization and the relative entropy
minimization. Equation (\ref{BGScanonical}) is essentially the
same as (\ref{boltzmannDistribution}).

For the Burg entropy $h'(x)=-\frac{1}{x}$, $g(x)=-\frac{1}{x}$ too
and
\begin{equation}\label{BurgDistr}
p_i=-\frac{p_i^*}{\mu_0+u_i \mu}
\end{equation}
 For the Lagrange multipliers $\mu_0, \mu$ we have a system of two
 algebraic equations
\begin{equation}\label{Burgcanonical}
\sum_{i=1}^n\frac{p_i^*}{\mu_0+u_i \mu}=-1, \,
\sum_{i=1}^n\frac{p_i^*u_i}{\mu_0+u_i \mu}=-U
\end{equation}

For the convex combination of the BGS and Burg entropies
$h'(x)=\beta \ln x - \frac{1-\beta}{x}$ ($0<\beta<1$), and the
function $x=g(y)$ is a solution of a transcendent equation
\begin{equation}\label{CCinverse}
\beta \ln x - \frac{1-\beta}{x}=y
\end{equation}
Such a solution exists for all real $y$ because this $h'(x)$ is a
(monotonic) bijection of $]0,\infty[$ on the real line.

Solution to Equation (\ref{CCinverse}) can be represented through a
special function, the Lambert function \cite{Lambert}. This function
is a solution to the transcendent equation $$w{\rm e}^ w=z$$ and is
also known as $W$ function, $\Omega$ function or modified logarithm
${\rm lm}z$ \cite{ENTR2}. Below we use the main branch $w={\rm lm}z$
for which ${\rm lm}z>0$ if $z>0$ and ${\rm lm}0=0$. Let us write
(\ref{CCinverse}) in the form
\begin{equation}\label{CCinverse1}
\ln x - \frac{\delta}{x}=-\Lambda
\end{equation}
where $\delta=(1-\beta)/\beta$, $\Lambda=-y/\beta$. Then $$ x={\rm
e}^{-\Lambda}{\rm e}^{{\rm lm}(\delta {\rm e}^{\Lambda}) }$$
Another equivalent representation of the solution gives
$$x=\frac{\delta}{{\rm lm}(\delta {\rm e}^{\Lambda})}$$ Indeed,
let us take $z=\delta/x$ and calculate exponent of both sides of
(\ref{CCinverse1}). After simple transformations, we obtain $z{\rm
e}^z=\delta {\rm e}^{\Lambda}$.

The identity ${\rm lm}a=\ln a -\ln {\rm lm}a$ is convenient for
algebraic operations with this function. Many other important
properties are collected in \cite{Lambert}.

The generalized canonical distribution for the convex combination
of the BGS and Burg divergence is \cite{ENTR2}
\begin{equation}\label{CCgcd}
p_i=p_i^* {\rm e}^{-\Lambda_i}{\rm e}^{{\rm lm}(\delta {\rm
e}^{\Lambda_i}) }=\frac{\delta p^*_i}{{\rm lm}(\delta {\rm
e}^{\Lambda_i})}
\end{equation}
where $\Lambda_i=-\frac{1}{\beta}(\mu_0+u_i \mu)$,
$\delta=(1-\beta)/\beta$ and equations (\ref{LagrangeEq}) hold for
the Lagrange multipliers.

For small $1-\beta$ (small addition of the Burg entropy to the BGS
entropy) we have $$ p_i=p_i^* \left({\rm
e}^{-\Lambda_i}+\frac{1-\beta}{\beta}-\frac{(1-\beta)^2}{2\beta^2}{\rm
e}^{\Lambda_i}\right) + o((1-\beta)^2) $$ For the CR family
$h(x)=\frac{x(x^{\lambda}-1)}{\lambda (\lambda+1)}$,
$h'(x)=\frac{(\lambda+1) x^{\lambda}-1}{\lambda (\lambda+1)}$,
$g(x)=(\frac{\lambda (\lambda+1)x
+1}{(\lambda+1)})^{\frac{1}{\lambda}}$ and

\begin{equation}\label{CRcanondistr}
p_i=p_i^*\left(\frac{\lambda (\lambda+1)(\mu_0+u_i \mu)
+1}{(\lambda+1)}\right)^{\frac{1}{\lambda}}
\end{equation}

For $\lambda = 1$ (a quadratic divergence) we easily get linear
equations and explicit solutions for $\mu_0$ and $\mu$. If $\lambda
=\frac{1}{2}$ then equations for the Lagrange multipliers
(\ref{LagrangeEq}) become quadratic and also allow explicit
solution. The same is true for $\lambda =\frac{1}{3}$ and
$\frac{1}{4}$ but explicit solutions to the correspondent cubic or
quartic equations are too cumbersome.

We studied the generalized canonical distributions for one
condition $u(P)=U$ and main families of entropies.  For the BGS
entropy, the method of Lagrange multipliers gives one transcendent
equation for the multiplier $\mu_1$ and explicit expression for
$\mu_0$ as a function of $\mu_1$ (\ref{BGScanonical}). In general,
for functions $H_h$, the method gives a system of two equations
(\ref{LagrangeEq}). For the Burg entropy this is a system of
algebraic equation (\ref{Burgcanonical}). For a convex combination
of the BGS and the Burg entropies the expression for generalized
canonical distribution function includes the special Lambert
function (\ref{CCgcd}). For the CR family the generalized
canonical distribution is presented by formula
(\ref{CRcanondistr}). for several values of $\lambda$ it can be
represented in explicit form. The Tsallis entropy family is a
subset of the CR family (up to constant multipliers).

\subsection{Polyhedron of Generalized Canonical Distributions for the Markov Order}

The set of the most random distributions with respect to the Markov
order under given condition consists of those distributions which
may be achieved by randomization  which has the given equilibrium
distribution and does not violate the condition.

In the previous section, this set was characterized for a single
condition $\sum_i p_i u_i = U$, $U\neq U^*$  by a system of
inequalities and equations (\ref{1condExtrPoints}). It is a
polyhedron that is an intersection of the closure of one compartment
with the hyperplane of condition. Here we construct the dual
description of this polyhedron as a convex envelope of the set of
extreme points (vertices).

The Krein--Milman theorem gives  general backgrounds of such a
representation of convex compact sets in locally convex
topological vector spaces \cite{EdwardsKreinMilman1995}: a compact
convex set is the closed convex hull of its extreme points. (An
extreme point of a convex set $K$ is a point $x\in K$ which cannot
be represented as an average $x=\frac{1}{2}(y+z)$ for $y,z \in K$,
$y,z\neq x$.)

Let us assume that there are $k+1\leq n$ different numbers in the
set of numbers $u_i (U-U^*)$. There exists the unique surjection
$\sigma: \{1,2,\ldots n\} \to \{1,2,\ldots k+1\}$ with the
following properties: $\sigma(i) < \sigma(j)$ if and only if $u_i
(U-U^*) >u_j (U-U^*)$ (hence, $\sigma(i) = \sigma(j)$ if and only
if $u_i (U-U^*) =u_j (U-U^*)$). The polyhedron of generalized
canonical distributions is the intersection of the condition plane
$\sum_i p_i u_i = U$ with the closure of $\mathcal{C}_{\sigma}$.

This closure is a simplex with vertices $P^*$ and $v_r$
($r=1,2,\ldots k$) (\ref{verticesC}). The vertices of the
intersection of this simplex with the condition hyperplane belong to
edges of the simplex, hence we can easily find all of them: the edge
$[x,y]$ has nonempty intersection with the condition hyperplane if
either $u(x)\geq U \& u(y)\leq U$ or $u(x)\leq U \& u(y)\geq U$.
This intersection is a single point $P$ if $u(x)\neq u(y)$:
\begin{equation}\label{vertexEdge}
P= \lambda x +(1-\lambda) y, \;\; \lambda=\frac{u(y)-U}{u(y)-u(x)}
\end{equation}
If $u(x)= u(y)$ then the intersection is the whole edge, and the
vertices are $x$ and $y$.

For example, if $U$ is sufficiently close to $U^*$ then the
intersection is a simplex with $k$ vertices $w_r$ ($r=1,2,\ldots
k$). Each $w_r$ is the intersection of the edge $[P^*,v_r]$ with
the condition hyperplane.

Let us find these vertices explicitly. We have a system of two
equations
\begin{equation}
\begin{split}
& a \sum_{i, \, \sigma(i)\leq r} p^*_i+b\sum_{i, \, \sigma(i)>
r}p^*_i=1 \, ; \\
 & a\sum_{i, \, \sigma(i)\leq r} u_i
p^*_i+b\sum_{i, \, \sigma(i)> r} u_ip^*_i=U
\end{split}
\end{equation}
Position of the vertex $w_r$ on the  edge $[P^*,v_r]$  is given by
the following expressions
\begin{equation}\label{verticesGCD}
\begin{split}
&\frac{p_i}{p_i^*} =a, \; \; {\rm for} \;\;
 \sigma(i) \leq r \, ; \;\;
 \frac{p_i}{p_i^*}=b, \; \; {\rm for} \;\;
 \sigma(i) > r \\
 & a=1+\frac{(U-U^*)\sum_{i, \, \sigma(i)> r}p^*_i}{\sum_{i, \,
\sigma(i)> r}p^*_i\sum_{i, \, \sigma(i)\leq r} u_i p^*_i-\sum_{i,
\, \sigma(i)\leq r} p^*_i \sum_{i, \, \sigma(i)> r} u_ip^*_i } \\
 & b=1-\frac{(U-U^*)\sum_{i, \, \sigma(i)\leq r} p^*_i }{\sum_{i, \,
\sigma(i)> r}p^*_i\sum_{i, \, \sigma(i)\leq r} u_i p^*_i-\sum_{i,
\, \sigma(i)\leq r} p^*_i \sum_{i, \, \sigma(i)> r} u_ip^*_i }
\end{split}
\end{equation}
If $b \geq 0$ for all $r$ then the polyhedron of generalized
canonical distributions is a simplex with vertices $w_r$. If the
solution becomes negative for some $r$ then the set of vertices
changes qualitatively and some of them belong to the base of
$\mathcal{C}_{\sigma}$. For example, in
Figure~\ref{CondExtrdDifferent}a the interval of the generalized
canonical distribution (1D polyhedron) has vertices of two types:
one belongs to the lateral face, another is situated on the
basement of the compartment. In Figure~\ref{CondExtrdDifferent}b
both vertices belong to the lateral faces.

Vertices $w_r$ on the edges $[P^*,v_r]$ have very special
structure: the ratio $p_i/p_i^*$ can take for them only two
values, it is either  $a$ or $b$.

Another form for representation of vertices $w_r$
(\ref{verticesGCD}) can be found as follows. $w_r$ belongs to the
edge $[P^*,v_r]$, hence, $w_r=\lambda P^*+(1-\lambda)v_r$ for some
$\lambda\in [0,1]$. Equation for the value of $\lambda$ follows
from the condition $u(w_r)=U$: $\lambda U^*+(1-\lambda)u(v_r)=U$.
Hence, we can use (\ref{vertexEdge}) with $x=P^*$, $y=v_r$.

For sufficiently large value of $U-U^*$ for some of these vertices
$b$ loses positivity, and instead of them the vertices on edges
$[v_r,v_q]$ (\ref{verticesC}) appear.

There exists a vertex on the edge $[v_r,v_q]$ if either
$u(v_r)\geq U \& u(v_q)\leq U$ or $u(v_r)\leq U \& u(v_q)\geq U$.
If $u(v_r)\neq u(v_q)$ then his vertex has the form $P=\lambda v_r
+ (1-\lambda) v_q$ and for $\lambda$ the condition $u(P)=U$ gives
(\ref{vertexEdge}) with $x=v_r$, $y=v_q$. If $u(v_r)= u(v_q)$ then
the edge $[u(v_r),u(v_q)]$ belongs to the condition plane and the
extreme distributions are $u(v_r)$ $u(v_q)$.

For each of $v_r$ the ratio $p_i/p_i^*$ can take only two values:
$a_r$ or 0. Without loss of generality we can assume that $q>r$.
For a convex combination $\lambda v_r + (1-\lambda) v_q$
($1>\lambda>0$) the ratio $p_i/p_i^*$ can take three values:
$\lambda a_r + (1-\lambda) a_q$ (for $ \sigma(i) \leq r $),
$(1-\lambda) a_q$ (for $r < \sigma(i) \leq q$)  and 0 (for
$\sigma(i) > q$).

The case when a vertex is one of the $v_r$ is also possible. In
this case, there are two possible values of $p_i/p^*_i$, it is
either $a_r$ or $0$.

All the generalized canonical distributions from the polyhedron
are convex combinations of its extreme points (vertices). If the
set of vertices is $\{w_r\}$, then for any generalized canonical
distributions $P=\sum \lambda_i w_i$ ($\lambda_i\geq 0$, $\sum_i
\lambda_i=1$). The vertices can be found explicitly. Explicit
formulas for the extreme generalized canonical distributions are
given in this section: (\ref{verticesGCD}) and various
applications of (\ref{vertexEdge}). These formulas are based on
the description of compartment $\mathcal{C}_{\sigma}$ given in
Proposition~\ref{ComartFaces} and Equation~(\ref{verticesC}).

\section{History of the Markov Order}
\subsection{Continuous Time Kinetics \label{SecHistory}}

We have to discuss the history of the Markov order in the wider
context of orders, with respect to which the solutions of kinetic
equations change monotonically in time. The Markov order is a nice
and constructive example of such an order and  at the same time the
prototype of all of them (similarly the Master Equation is a simple
example of kinetic equations and, at the same time, the prototype of
all kinetic equations).

The idea of orders and attainable domains (the lower cones of
these orders) in phase space was developed in many applications:
from biological kinetics to chemical kinetics and engineering. A
kinetic model includes information of various levels of detail and
of variable reliability. Several types of building block are used
to construct a kinetic model. The system of these building blocks
can be described, for example, as follows:
\begin{enumerate}
\item{The list of components (in chemical kinetics) or populations (in
mathematical ecology) or states (for general Markov chains);}
 \item{The list of elementary processes (the reaction mechanism, the graph of
 trophic  interactions or the transition graph),
 which is often supplemented by the lines or surfaces
  of partial equilibria of elementary  processes;}
 \item{The reaction rates and kinetic constants.}
\end{enumerate}

We believe that the lower level information is more accurate and
reliable: we know the list of component better than the mechanism of
transitions, and our knowledge of equilibrium surfaces is better
than the information about exact values of kinetic constants.

It is attractive to use the more reliable lower level information
for qualitative and quantitative study of kinetics. Perhaps, the
first example of such a analysis was performed in biological
kinetics.

In 1936, A.N.~Kolmogorov \cite{Kolmogorov1936} studied the
dynamics of a pair of interacting populations of prey ($x$) and
predator ($y$) in general form:
 $$\dot{x}=xS(x,y), \;\; \dot{y}=yW(x,y)$$
under monotonicity conditions: $\partial S(x,y)/\partial y <0$,
$\partial W(x,y)/\partial y <0$. The zero isoclines, the lines at
which the rate of change for one population is zero (given by
equations $S(x,y)=0$ or $W(x,y)=0$), are graphs of two functions
$y(x)$. These isoclines divide the phase space into compartments
(generically with curvilinear borders). In every compartment the
angle of possible directions of motion is given (compare to
Figure~\ref{3stateCones}).

Analysis of motion in these angles gives information about
dynamics without an exact knowledge of the kinetic constants. The
geometry of the zero isoclines intersection together with some
monotonicity conditions give important information about the
system dynamics \cite{Kolmogorov1936} without exact knowledge of
the right hand sides of the kinetic equations.

This approach to population dynamics was further developed by many
authors and applied to various problems
\cite{MayLeonard1975,Bazykin1998}. The impact of this work on
population dynamics was analyzed by K. Sigmund in review
\cite{Sigmung2007}.

It seems very attractive to use an attainable region instead of
the single trajectory in situations with incomplete information or
with information with different levels of reliability. Such
situations are typical in many areas of engineering. In 1964,
F.~Horn proposed to analyze the attainable regions for chemical
reactors \cite{Horn1964}. This approach was applied both to linear
and nonlinear kinetic equations and became popular in chemical
engineering. It was applied to the optimization of steady flow
reactors \cite{Glasser1987}, to batch reactor optimization by use
of tendency models without knowledge of detailed kinetics
\cite{Filippi-Bossy1989} and for optimization of the reactor
structure \cite{Hildebrandt1990}. Analysis of attainable regions
is recognized as a special geometric approach to reactor
optimization \cite{Feinberg1997} and as a crucially important part
of the new paradigm of chemical engineering \cite{Hill2009}.
Plenty of particular applications was developed: from
polymerization \cite{SmithMalone1997} to particle breakage in a
ball mill \cite{Metzger2009}. Mathematical methods for study of
attainable regions vary from the Pontryagin's maximum principle
\cite{McGregor1999} to linear programming \cite{Kauchali2002}, the
Shrink-Wrap algorithm \cite{Manousiouthakis2004} and convex
analysis.

The connection between attainable regions, thermodynamics and
stoichiometric reaction mechanisms was studied by A.N. Gorban in the
1970s. In 1979, he demonstrated how to utilize the knowledge about
partial equilibria of elementary processes to construct the
attainable regions \cite{Gorban1979}.

He noticed that the set (a cone) of possible direction for kinetics
is defined by thermodynamics and the reaction mechanism (the system
of the stoichiometric equation of elementary reactions).

Thermodynamic data are more robust than the reaction mechanism and
the reaction rates are known with lower accuracy than the
stoichiometry of elementary reactions. Hence, there are two types of
attainable regions. The first is the thermodynamic one, which use
the linear restrictions and the thermodynamic functions
\cite{GorbanChMMS1979}. The second is generated by thermodynamics
and stoichiometric equations of elementary steps (but without
reaction rates) \cite{Gorban1979,GorbanBYa1980}.

It was demonstrated that the attainable regions significantly
depend on the transition mechanism (Figure~\ref{RKCLfig}) and it is
possible to use them for the mechanisms discrimination
\cite{GorbanYa1980}.

Already simple examples demonstrate that the sets of distributions
which are accessible from a given initial distribution by Markov
processes with equilibrium are, in general,  {\em non-convex}
polytopes \cite{Gorban1979,Zylka1985} (see, for example, the
outlined region in Figure~\ref{RKCLfig}, or, for particular graphs
of transitions, any of the shaded regions there). This non-convexity
makes the analysis of attainability for continuous time Markov
processes more difficult (and also more intriguing).

\begin{figure}[t]
\caption{\leftskip=1cm \rightskip=1cm \label{RKCLfig}Attainable regions from an initial
distribution $a_0$ for a linear system with three components
$A_1,A_2,A_3$ in coordinates $c_1,c_2$ (concentrations of
$A_1,A_2$) ($c_3={\rm const}-c_1-c_2$) \cite{Gorban1979}: for a
full mechanism $A_1\rightleftarrows A_2\rightleftarrows
A_3\rightleftarrows A_1$ (outlined region), for a two-step
mechanism $A_1\rightleftarrows A_2$, $A_1\rightleftarrows A_3$
(horizontally shaded region) and for a two-step mechanism
$A_1\rightleftarrows A_2$, $A_2\rightleftarrows A_3$ (vertically
shaded region). Equilibrium is $a^*$. The dashed lines are partial
equilibria. }}
\centering{
\includegraphics[width=90mm]{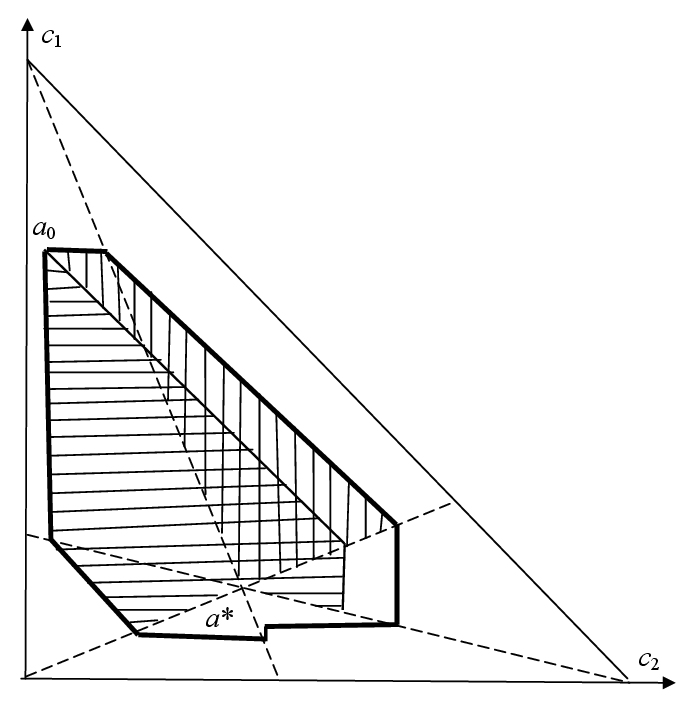}
\end{figure}

This approach was developed for all thermodynamic potentials and
for open systems as well \cite{G11984}. Partially, the results are
summarized in \cite{YBGE,GorbKagan2006}.

This approach was rediscovered by F.J. Krambeck \cite{Krambeck1984}
for linear systems, that is, for Markov chains, and by R. Shinnar
and other authors \cite{Shinnar1985} for more general nonlinear
kinetics. There was even an open discussion about priority
\cite{Bykov1987}. Now this geometric approach is applied to various
chemical and industrial processes.

\subsection{Discrete Time Kinetics}

In our paper we deal mostly with continuous time Markov chains. For
the discrete time Markov chains, the attainable regions  have two
important properties: they are convex and symmetric with respect to
permutations of states. Because of this symmetry and convexity, the
discrete time Markov order is characterized in detail. As far as we
can go in history, this work was begun in early 1970s by A. Uhlmann
and P.M. Alberti. The results of the first 10 years of this work
were summarized in monograph \cite{AlbertiUhlmann1982}. A more
recent bibliography (more than 100 references) is collected in
review \cite{AlbertiCUZ2008}.

This series of work was concentrated mostly on processes with
uniform equilibrium (doubly stochastic maps). The relative
majorization, which we also use in Section~\ref{MarOd}, and the
Markov order with respect to a non-uniform equilibrium was
introduced by P. Harremo{\"e}s in 2004 \cite{Harremo2004}. He used
formalism based on the Lorenz diagrams.

\section{Conclusion\label{Conclusion}}

Is playing with non-classical entropies and divergences just an
extension to the fitting possibilities (no sense---just fitting)? We
are sure now that this is not the case: two one-parametric families
of non-classical divergences are distinguished by the very natural
properties:
\begin{enumerate}
\item{They are Lyapunov functions for all Markov chains;}
\item{They become additive with respect to the joining of independent
systems after a monotone transformation of
scale;}
\item{They become additive with respect to a partitioning  of the
state space after a monotone transformation of scale.}
\end{enumerate}

Two families of smooth divergences (for positive distributions)
satisfy these requirements: the Cre\-ssie--Read family
\cite{CR1984,ReadCreass1988}
 $$H_{{\rm CR}\ \lambda}(P \| P^*)=\frac{1}{\lambda (\lambda+1)}
 \sum_i p_i\left[ \left(\frac{p_i}{p_i^*} \right)^{\lambda} -1
 \right]\ , \ \ \lambda\in ]-\infty,\infty[ $$
 and the convex combination of the Burg and Shannon relative
 entropies
 \cite{G11984,ENTR1}:
$$H(P \| P^*)= \sum_i (\beta p_i- (1-\beta)p_i^* )\log \left(
\frac{p_i}{p_i^*}\right)\ , \ \ \beta\in [0,1]$$

If we relax the differentiability property, then we have to add to
the the CR family two limiting cases: $$H_{\rm CR \ \infty}
(P\|P^*)= \max_i\left\{\frac{p_i}{p_i^*}\right\}-1 \ ; $$

$$H_{{\rm CR \ -\infty}} (P\|P^*)=
\max_i\left\{\frac{p_i^*}{p_i}\right\}-1 $$

Beyond these two distinguished one-parametric families there is the
whole world of the Csisz\'ar--Morimoto Lyapunov functionals for the
Master equation (\ref{Morimoto}). These functions monotonically
decrease along any solution of the Master equation. The set of all
these functions can be used to reduce the uncertainty by conditional
minimization: for each $h$ we could find a conditional minimizer of
$H_h(p)$.

Most users prefer to have an unambiguous choice of entropy: it would
be nice to have ``the best entropy" for any class of problems. But
from a certain point of view, ambiguity of the entropy choice is
unavoidable, and the choice of all conditional optimizers instead of
a particular one is a possible way to avoid an arbitrary choice. The
set of these minimizers evaluates the possible position of a
``maximally random" probability distribution. For many MaxEnt
problems the natural solution is not a fixed distribution, but a
well defined set of distributions.

The task to minimize functions $H_h(p)$ which depend on a functional
parameter $h$ seems too complicated. The Markov order gives us
another way for the evaluation of the set of possible ``maximally
random" probability distribution, and this evaluation is, in some
sense, the best one. We defined the Markov order, studied its
properties and demonstrated how it can be used to reduce
uncertainty.

It is quite surprising that the Markov order is generated by the
reversible Markov chains which satisfy the detailed balance
principle. We did not include any reversibility assumptions and
studied the general Markov chains. There remain some questions about
the structure and full description of the global Markov order.
Nevertheless, to find the set of conditionally extreme (``most
random") probability distributions, we need the local Markov order
only. This local order is fully described in Section~\ref{order} and
has a very clear geometric structure. For a given equilibrium
distribution $P^*$, the simplex of probability distributions is
divided by $n(n-1)/2$ hyperplanes of ``partial equilibria" (this
terminology comes from chemical kinetics
\cite{Gorban1979,GorbKagan2006}): $\frac{p_i}{p_i^*}=
\frac{p_j}{p_j^*}$ (there is one hyperplane for each pair of states
$(i,j)$). In each compartment a cone of all possible time
derivatives of the probability distribution is defined as a conic
envelope of $n(n-1)/2$ vectors (\ref{LocalOrderCone}). The extreme
rays of this cone are explicitly described in \linebreak Proposition
\ref{extremerays6} (\ref{ConeDescr}). This cone defines the local
Markov order. When we look for conditionally extreme distributions,
this cone plays the same role as a hyperplane given by entropy
growth condition \linebreak ( $\D S / \D t > 0$) in the standard
approach.

For the problem of the generalized canonical (or reference)
distribution the Markov order gives a polyhedron of the extremely
disordered distributions. The vertices of that polyhedron can be
computed explicitly.

The construction of efficient algorithms for numerical calculation
of conditionally extreme compacts in high dimensions is a
challenging task for our future work as well as the application of
this methodology to real life problems.

\section*{Acknowledgements}

Suggestions from Mike George, Marian Grendar, Ivan Tyukin and
anonymous referees are gratefully acknowledged.

\bibliographystyle{mdpi}
\makeatletter
\renewcommand\@biblabel[1]{#1. }
\makeatother

\section*{Appendix \label{App}}

\noindent{\it Proof of Theorem \ref{Theorem1NoMore}}. The problem
is to find all such universal and trace--form Lyapunov functions
$H$ for Markov chains, that there exists a monotonous function
$F$, such that $F(H({P}))=F(H({Q}))+F(H({R}))$ if ${
P}=p_{ij}=q_{i}r_{j}$.

With Lemma 1 we get that
$$H({P})=\sum_{i,j}q_{i}^{*}r_{j}^{*}h\left(\frac{q_{i}r_{j}}{q_{i}^{*}r_{j}^{*}}\right),
\;\;
H({Q})=\sum_{i}q_{i}^{*}h\left(\frac{q_{i}}{q_{i}^{*}}\right),
\;\; H({R})
=\sum_{j}r_{j}^{*}h\left(\frac{r_{j}}{r_{j}^{*}}\right)$$
 Let $F(x)$ and $h(x)$ be differentiable as many times as needed.
Differentiating the equality $F(H({P}))=F(H({Q}))+F(H({R}))$ on
$r_{1}$ and $q_{1}$ taking into account that
$q_{n}=1-\sum_{i=1}^{n-1}q_{i}$ and
$r_{m}=1-\sum_{j=1}^{m-1}r_{j}$ we get that
$F'(H({P}))H''_{q_{1}r_{1}}({P})=-F''(H({P}))H'_{q_{1}}({P})H'_{r_{1}}({P})$,
or, if $-\frac{F'(H({P}))}{F''(H({P}))}=G(H({P}))$ then
\begin{equation}
G(H({P}))=\frac{H'_{q_{1}}({P})H'_{r_{1}}({P})}{H''_{q_{1}r_{1}}({P})}
\end{equation}
It is possible if and only if every linear differential operator
of the first order, which annulates $H({P})$ and $\sum p_{i}$,
annulates also
\begin{equation}\label{Geqn}
\frac{H'_{q_{1}}({P})H'_{r_{1}}({P})}{H''_{q_{1}r_{1}}({P})}
\end{equation}
and it means that every differential operator which has the form
\begin{equation}\label{Dfrm}
D=\left(\frac{\partial H({P})}{\partial q_{\gamma}}-\frac{\partial
H({P})}{\partial q_{\alpha}}\right)\frac{\partial}{\partial
q_{\beta}}+\left(\frac{\partial H({P})}{\partial
q_{\beta}}-\frac{\partial H({P})}{\partial
q_{\gamma}}\right)\frac{\partial}{\partial
q_{\alpha}}+\left(\frac{\partial H({P})}{\partial
q_{\alpha}}-\frac{\partial H({P})}{\partial
q_{\beta}}\right)\frac{\partial}{\partial q_{\gamma}}
\end{equation}
annulates (\ref{Geqn}). For $\beta=2, \alpha=3, \gamma=4$  we get
the following equation
\begin{eqnarray}\label{Diffur}
\nonumber &
F_{1}({Q},{R})\left[h'\left(\frac{q_{2}r_{1}}{q_{2}^{*}r_{1}^{*}}\right)-h'\left(\frac{q_{2}r_{m}}{q_{2}^{*}r_{m}^{*}}\right)+
\frac{q_{2}r_{1}}{q_{2}^{*}r_{1}^{*}}h''\left(\frac{q_{2}r_{1}}{q_{2}^{*}r_{1}^{*}}\right)-
\frac{q_{2}r_{m}}{q_{2}^{*}r_{m}^{*}}h''\left(\frac{q_{2}r_{m}}{q_{2}^{*}r_{m}^{*}}\right)\right]+\\
\nonumber &
F_{2}({Q},{R})\left[h'\left(\frac{q_{3}r_{1}}{q_{3}^{*}r_{1}^{*}}\right)-h'\left(\frac{q_{3}r_{m}}{q_{3}^{*}r_{m}^{*}}\right)+
\frac{q_{3}r_{1}}{q_{3}^{*}r_{1}^{*}}h''\left(\frac{q_{3}r_{1}}{q_{3}^{*}r_{1}^{*}}\right)-
\frac{q_{3}r_{m}}{q_{3}^{*}r_{m}^{*}}h''\left(\frac{q_{3}r_{m}}{q_{3}^{*}r_{m}^{*}}\right)\right]+\\
\nonumber &
F_{3}({Q},{R})\left[h'\left(\frac{q_{4}r_{1}}{q_{4}^{*}r_{1}^{*}}\right)-h'\left(\frac{q_{4}r_{m}}{q_{4}^{*}r_{m}^{*}}\right)+
\frac{q_{4}r_{1}}{q_{4}^{*}r_{1}^{*}}h''\left(\frac{q_{4}r_{1}}{q_{4}^{*}r_{1}^{*}}\right)-
\frac{q_{4}r_{m}}{q_{4}^{*}r_{m}^{*}}h''\left(\frac{q_{4}r_{m}}{q_{4}^{*}r_{m}^{*}}\right)\right]=0\\
\end{eqnarray}

where

\begin{eqnarray}
\nonumber
F_{1}({Q},{R})=\sum_{j}r_{j}\left[h'\left(\frac{q_{4}r_{j}}{q_{4}^{*}r_{j}^{*}}\right)-h'\left(\frac{q_{3}r_{j}}{q_{3}^{*}r_{j}^{*}}\right)\right]\,
;\\ \nonumber
F_{2}({Q},{R})=\sum_{j}r_{j}\left[h'\left(\frac{q_{2}r_{j}}{q_{2}^{*}r_{j}^{*}}\right)-h'\left(\frac{q_{4}r_{j}}{q_{4}^{*}r_{j}^{*}}\right)\right]\,
;\\ \nonumber
F_{3}({Q},{R})=\sum_{j}r_{j}\left[h'\left(\frac{q_{3}r_{j}}{q_{3}^{*}r_{j}^{*}}\right)-h'\left(\frac{q_{2}r_{j}}{q_{2}^{*}r_{j}^{*}}\right)\right]\\
\nonumber
\end{eqnarray}

If we apply the differential operator $\frac{\partial}{\partial
r_{2}}-\frac{\partial}{\partial r_{3}}$, which annulates the
conservation law $\sum_{j}r_{j}=1$, to the left part of
(\ref{Diffur}), and denote $f(x)=xh''(x)+h'(x)$,
$x_{1}=\frac{q_{2}}{q_{2}^{*}}$, $x_{2}=\frac{q_{3}}{q_{3}^{*}}$,
$x_{3}=\frac{q_{4}}{q_{4}^{*}}$, $y_{1}=\frac{r_{1}}{r_{1}^{*}}$,
$y_{2}=\frac{r_{m}}{r_{m}^{*}}$, $y_{3}=\frac{r_{2}}{r_{2}^{*}}$,
$y_{4}=\frac{r_{3}}{r_{3}^{*}}$, we get the equation
\begin{eqnarray}\label{FirstFeq}
\nonumber
(f(x_{3}y_{3})-f(x_{2}y_{3})-f(x_{3}y_{4})+f(x_{2}y_{4}))(f(x_{1}y_{1})-f(x_{1}y_{2}))+\\
(f(x_{1}y_{3})-f(x_{3}y_{3})-f(x_{1}y_{4})+f(x_{3}y_{4}))(f(x_{2}y_{1})-f(x_{2}y_{2}))+\\
\nonumber
(f(x_{2}y_{3})-f(x_{1}y_{3})-f(x_{2}y_{4})+f(x_{1}y_{4}))(f(x_{3}y_{1})-f(x_{3}y_{2}))=0
\end{eqnarray}
or, after differentiation on $y_{1}$ and $y_{3}$ and denotation
$g(x)=f'(x)$
\begin{eqnarray}\label{SecFeq}
&
&x_{1}g(x_{1}y_{1})(x_{3}g(x_{3}y_{3})-x_{2}g(x_{2}y_{3}))+x_{2}g(x_{2}y_{1})(x_{1}g(x_{1}y_{3})-\\
\nonumber & &-x_{3}g(x_{3}y_{3}))+
x_{3}g(x_{3}y_{1})(x_{2}g(x_{2}y_{3})-x_{1}g(x_{1}y_{3}))=0
\end{eqnarray}
If $y_{3}=1$, $y_{1}\neq0$, $\varphi(x)=xg(x)$, we get after
multiplication (\ref{SecFeq}) on $y_{1}$
\begin{equation}\label{TrdFeq}
\varphi(x_{1}y_{1})(\varphi(x_{3})-\varphi(x_{2}))+\varphi(x_{2}y_{1})
(\varphi(x_{1})-\varphi(x_{3}))+\varphi(x_{3}y_{1})(\varphi(x_{2})-\varphi(x_{1}))=0
\end{equation}
It implies that for every three positive numbers $\alpha$,
$\beta$, $\gamma$ the functions $\varphi(\alpha x)$,
$\varphi(\beta x)$, $\varphi(\gamma x)$ are linearly dependent,
and for $\varphi(x)$ the differential equation
\begin{equation}
ax^{2}\varphi''(x)+bx\varphi'(x)+c\varphi(x)=0
\end{equation}
holds. This differential equation has solutions of two kinds:
\begin{enumerate}
\item{$\varphi(x)=C_{1}x^{k_{1}}+C_{2}x^{k_{2}}$, $k_{1}\neq k_{2}$,
$k_{1}$ and $k_{2}$ are real or complex-conjugate numbers.}
\item{$\varphi(x)=C_{1}x^{k}+C_{2}x^{k}\ln x$.}
\end{enumerate}

Let us check, which of these solutions satisfy the functional
equation (\ref{TrdFeq}).
\begin{enumerate}
\item{$\varphi(x)=C_{1}x^{k_{1}}+C_{2}x^{k_{2}}$. After substitution
of this into (\ref{TrdFeq}) and calculations we get
$$C_{1}C_{2}(y_{1}^{k_{1}}-y_{1}^{k_{2}})(x_{1}^{k_{1}}x_{3}^{k_{2}}-
x_{1}^{k_{1}}x_{2}^{k_{2}}+x_{1}^{k_{2}}x_{2}^{k_{1}}-
x_{2}^{k_{1}}x_{3}^{k_{2}}+x_{2}^{k_{2}}x_{3}^{k_{1}}-
x_{1}^{k_{2}}x_{3}^{k_{1}})=0 $$
 This means that $C_{1}=0$, or $C_{2}=0$, or $k_{1}=0$, or
$k_{2}=0$ and the solution of this kind can have only the form
$\varphi(x)=C_{1}x^{k}+C_{2}$.}
 \item{$\varphi(x)=C_{1}x^{k}+C_{2}x^{k}\ln x$. After substitution of
this into (\ref{TrdFeq}) and some calculations if $y_{1}\neq0$ we
get $$C_{2}^{2}((x_{1}^{k}-x_{2}^{k})x_{3}^{k}\ln
x_{3}+(x_{3}^{k}-x_{1}^{k})x_{2}^{k}\ln
x_{2}+(x_{2}^{k}-x_{3}^{k})x_{1}^{k}\ln x_{1})=0 $$ This means
that either $C_{2}=0$ and the solution is $\varphi(x)=C_{1}x^{k}$
or $k=0$ and the solution is $\varphi(x)=C_{1}+C_{2}\ln x$.}
\end{enumerate}
So, the equation (\ref{TrdFeq}) has two kinds of solutions:
\begin{enumerate}
\item{$\varphi(x)=C_{1}x^{k}+C_{2}$,}
\item{$\varphi(x)=C_{1}+C_{2}\ln x$}
\end{enumerate}

Let us solve the equation $f(x)=xh''(x)+h'(x)$ for each of these
two cases.
\begin{enumerate}
\item{$\varphi(x)=C_{1}x^{k}+C_{2}$,
$g(x)=C_{1}x^{k-1}+\frac{C_{2}}{x}$, there are two possibilities:}
\begin{enumerate}
\item[1.1)]{$k=0$. Then $g(x)=\frac{C}{x}$, $f(x)=C\ln x+C_{1}$,
$h(x)=C_{1}x\ln x+C_{2}\ln x+C_{3}x+C_{4}$;}
 \item[1.2)]{$k\neq0$. Then
$f(x)=Cx^{k}+C_{1}\ln x+C_{2}$, and here are also two
possibilities:}
 \begin{enumerate}
 \item[1.2.1)]{$k=-1$. Then $h(x)=C_{1}\ln^{2}x+C_{2}x\ln x+C_{3}\ln
x+C_{4}x+C_{5}$;}
 \item[1.2.2)]{$k\neq-1$. Then $h(x)=C_{1}x^{k+1}+C_{2}x\ln x+C_{3}\ln
x+C_{4}x+C_{5}$;}
 \end{enumerate}
   \end{enumerate}
\item{$\varphi(x)=C_{1}+C_{2}\ln x$; $g(x)=C_{1}\frac{\ln
x}{x}+\frac{C_{2}}{x}$; $f(x)=C_{1}\ln^{2}x+C_{2}\ln x+C_{3}$;
$h(x)=C_{1}x\ln^{2}x+C_{2}x\ln x+C_{3}\ln x+C_{4}x+C_{5}$.}
\end{enumerate}
(We have renamed constants during the calculations).

For the next step let us check, which of these solutions remains a
solution to equation (\ref{Diffur}). The result is that there are
just two families of functions $h(x)$ such, that equation
(\ref{Diffur}) holds:
\begin{enumerate}
\item{$h(x)=Cx^{k}+C_{1}x+C_{2}$, $k\neq0$, $k\neq1$,}
\item{$h(x)=C_{1}x\ln x+C_{2}\ln x+C_{3}x+C_{4}$.}
\end{enumerate}
The function $h(x)$ should be convex. This condition determines
the signs of coefficients $C_{i}$.

The corresponding divergence $H(P\| P^*)$ is either one of the CR
entropies or a convex combination of Shannon's and Burg's
entropies up to a monotonic transformation. ${ \mathbf{\square}}$

\vspace{2mm} \noindent {\em Characterization of Additive
Trace--form Lyapunov Functions for Markov Chains.} We will
consider three important properties of Lyapunov functions $H(P
\|P^*)$:
\begin{enumerate}
\item{{\it Universality}: $H$ is a Lyapunov function for Markov
chains (\ref{MAsterEq1}) with a given equilibrium $P^{*}$ for
every possible values of kinetic coefficients $k_{ij}\geq0$.}
\item{$H$ is a {\it trace--form function}.
\begin{equation}\label{trfrm}
H(P \|P^*)=\sum_{i}f(p_{i},p_{i}^{*})
\end{equation}
where $f$ is a differentiable function of two variables.}
 \item{$H$ is
{\it additive} for composition of independent subsystems. It means
that if ${P}=p_{ij}=q_{i}r_{j}$ and $
P^{*}=p_{ij}^{*}=q_{i}^{*}r_{j}^{*}$ then
$H(P\|P^*)=H(Q\|Q^*)+H(R\|R^*)$.}
\end{enumerate}
Here and further we suppose
$0<p_{i},p_{i}^{*},q_{i},q_{i}^{*},r_{i},r_{i}^{*}<1$.

We consider the additivity condition as a functional equation and
solve it. The following theorem describes all Lyapunov functions
for Markov chains, which have all three properties 1) - 3)
simultaneously.

Let $f(p,p^{*})$ be a twice differentiable function of two
variables.

\begin{theorem}\label{thAddTrace}If a function $H(P \|P^*)$ has all the
properties 1)-3) simultaneously, then
\begin{equation}\label{th1eq1}
f(p,p^{*})=p_{i}^{*}h\left(\frac{p}{p^{*}}\right),\; \; H(P
\|P^*)= \sum_{i}p_{i}^{*}h\left(\frac{p_{i}}{p_{i}^{*}}\right)
\end{equation}
where
\begin{equation}\label{th1eq2}
h(x)=C_{1}\ln x+C_{2}x\ln x,\mbox{ }C_{1}\leq0,\mbox{ }C_{2}\geq0
\end{equation}
\end{theorem}
\begin{proof}We follow here the P. Gorban proof \cite{ENTR3}.
Another proof of this theorem was proposed in \cite{Harr2007}. Due
to Lemma~\ref{MorimotoChar} let us take $H(P \|P^*)$ in the form
(\ref{th1eq1}). Let $h$  be twice differentiable in the interval
$]0,+\infty[$. The additivity equation
\begin{equation}\label{addeq}
H(P\|P^*)-H(Q\|Q^*)-H(R\|R^*)=0
\end{equation}
holds. Here (in (\ref{addeq}))
\begin{eqnarray*}
&&q_{n}=1-\sum_{i=1}^{n-1}q_{i}, \, r_{m}=1-\sum_{j=1}^{m-1}r_{j},
\,  P=p_{ij}=q_{i}r_{j} \\
&&H(P\|P^*)=\sum_{i,j}q_{i}^{*}r_{j}^{*}h\left(\frac{q_{i}r_{j}}{q_{i}^{*}r_{j}^{*}}\right),
\,
H(Q\|Q^*)=\sum_{i}q_{i}^{*}h\left(\frac{q_{i}}{q_{i}^{*}}\right),\,
H(R\|R^*)=\sum_{j}r_{j}^{*}h\left(\frac{r_{j}}{r_{j}^{*}}\right)
\end{eqnarray*}
Let us take the derivatives of this equation first on $q_{1}$ and
then on $r_{1}$. Then we get the equation ($g(x)=h'(x)$)
\begin{eqnarray*}
&g(\frac{q_{1}r_{1}}{q_{1}^{*}r_{1}^{*}})-g(\frac{q_{n}r_{1}}{q_{n}^{*}r_{1}^{*}})-
g(\frac{q_{1}r_{m}}{q_{1}^{*}r_{m}^{*}})+g(\frac{q_{n}r_{m}}{q_{n}^{*}r_{m}^{*}})+\\
&+\frac{q_{1}r_{1}}{q_{1}^{*}r_{1}^{*}}g'(\frac{q_{1}r_{1}}{q_{1}^{*}r_{1}^{*}})-
\frac{q_{n}r_{1}}{q_{n}^{*}r_{1}^{*}}g'(\frac{q_{n}r_{1}}{q_{n}^{*}r_{1}^{*}})-
\frac{q_{1}r_{m}}{q_{1}^{*}r_{m}^{*}}g'(\frac{q_{1}r_{m}}{q_{1}^{*}r_{m}^{*}})+
\frac{q_{n}r_{m}}{q_{n}^{*}r_{m}^{*}}g'(\frac{q_{n}r_{m}}{q_{n}^{*}r_{m}^{*}})=0
\end{eqnarray*}
Let us denote $x=\frac{q_{1}r_{1}}{q_{1}^{*}r_{1}^{*}}$,
$y=\frac{q_{n}r_{1}}{q_{n}^{*}r_{1}^{*}}$,
$z=\frac{q_{1}r_{m}}{q_{1}^{*}r_{m}^{*}}$, and
$\psi(x)=g(x)+xg'(x)$. It is obvious that if $n$ and $m$ are more
than 2, then $x$, $y$ and $z$ are independent and can take any
positive values. So, we get the functional equation:
\begin{equation}
\psi\left(\frac{yz}{x}\right)=\psi(y)+\psi(z)-\psi(x)
\end{equation}
Let's denote $C_{2}=-\psi(1)$ and
$\psi_{1}(\alpha)=\psi(\alpha)-\psi(1)$ and take $x=1$. We get
then
\begin{equation}
\psi_{1}(yz)=\psi_{1}(y)+\psi_{1}(z)
\end{equation}
the Cauchy functional equation \cite{Aczel1966}. The solution of
this equation in the class of measurable functions is
$\psi_{1}(\alpha)=C_{1}\ln\alpha$, where $C_{1}$ is constant. So
we get $\psi(x)=C_{1}\ln x+C_{2}$ and $g(x)+xg'(x)=C_{1}\ln
x+C_{2}$. The solution is $g(x)=\frac{C_{3}}{x}+C_{1}\ln
x+C_{2}-C_{1}$; $h(x)=\int(\frac{C_{3}}{x}+C_{1}\ln
x+C_{2}-C_{1})dx=C_{3}\ln x+C_{1}x\ln x+(C_{2}-2C_{1})x+C_{4}$,
or, renaming constants, $h(x)=C_{1}\ln x+C_{2}x\ln
x+C_{3}x+C_{4}$. In the expression for $h(x)$ there are two
parasite constants $C_{3}$ and $C_{4}$ which occurs because the
initial equation was differentiated twice. So, $C_{3}=0$,
$C_{4}=0$ and $h(x)=C_{1}\ln x+C_{2}x\ln x$. Because $h$ is
convex, we have $C_{1}\leq0$ and $C_{2}\geq0$.
\end{proof}

So, any universal additive trace--form Lyapunov function for
Markov chains is a convex combination of the BGS entropy and the
Burg entropy.

\end{document}